\documentclass[journal,twoside]{IEEEtran}

\usepackage{mathrsfs}
\usepackage{amsfonts}
\usepackage{amsmath,amssymb}
\usepackage{graphicx}
\usepackage{float}
\usepackage[dvips]{epsfig}
\usepackage[dvips]{color}

\def\dref#1{(\ref{#1})}

\begin{document}
\newtheorem{assumption}{Assumption}
\newtheorem{lemma}{Lemma}
\newtheorem{theorem}{Theorem}
\newtheorem{corollary}{Corollary}
\newtheorem{remark}{Remark}
%
\title{Fully Distributed Adaptive Output Feedback Protocols
for Linear Multi-Agent Systems with Directed Graphs: A Sequential Observer Design Approach}


\author{Yuezu Lv,
Zhongkui Li,
Zhisheng Duan, and Jie Chen
\thanks{This work was supported in part by the National Natural Science Foundation of China under grants 61473005, 11332001, 61225013,
and in part by the Hong Kong RGC under the project CityU 111810, CityU 111511.}
\thanks{Y. Lv, Z. Li, and Z. Duan are with the State Key Laboratory for Turbulence and Complex Systems, Department of Mechanics and Engineering Science, College of Engineering, Peking University, Beijing 100871, China. E-mail: {\tt\{yzlv,zhongkli,duanzs\}@pku.edu.cn}}
\thanks{J. Chen is with the Department of Electronic Engineering, City University of Hong Kong, Kowloon, Hong Kong.
E-mail: {\tt jichen@cityu.edu.hk}}

}

%



\IEEEtitleabstractindextext{%
\begin{abstract}
This paper studies output feedback consensus protocol design problems for linear multi-agent systems with directed graphs.
 We consider both leaderless and leader-follower consensus with a leader whose control input is nonzero and bounded.
We propose a novel sequential observer design approach, which makes it possible to design fully distributed adaptive output feedback protocols
that the existing methods fail to accomplish. With the sequential observer architecture,
we show that leaderless consensus can be achieved for any strongly connected directed graph
in a fully distributed manner, whenever the agents are stabilizable and detectable.
For the case with a leader of bounded control input,
we further present novel distributed adaptive output feedback protocols, which include
nonlinear functions to deal with the effect of the leaders's nonzero control input
and are able to achieve leader-follower consensus for any directed graph containing a directed spanning tree with the leader as the root.
\end{abstract}

}

\maketitle

\IEEEdisplaynontitleabstractindextext

%
\IEEEpeerreviewmaketitle

\section{Introduction}

Over the past decade, the consensus control problem of multi-agent systems has emerged as a focal research topic in the field of control, due to its various applications to, e.g., UAV formation flying, multi-point surveillance, and distributed reconfigurable sensor networks \cite{ren2007information,antonelli2013interconnected}.
Considerable work from different perspectives has been conducted on consensus and other related cooperative control problems; see the recent papers \cite{olfati-saber2004consensus,ren2005consensus,li2010consensus,ren2007information,cao2013overview,antonelli2013interconnected},
the monographs \cite{RenBeard07_Springer,li2014cooperative}, and the references therein.

Existing consensus algorithms can be essentially divided into two broad categories,
namely, consensus without a leader (i.e., leaderless consensus) and
consensus with a leader, whereas the latter ia also called leader-follower consensus
or distributed tracking.
In a leader-follower consensus problem, it is often the case that the leader may need to implement its own control
actions to achieve certain objectives, e.g., to reach a desirable consensus trajectory or to avoid hazardous
obstacles. Thus, compared to leaderless consensus,
an additional difficulty with leader-follower consensus lies in how to deal with the effect of
the leader's control input which is available to at most a small subset of the followers.

\subsection{Motivations and Related Works}

A central task in consensus studies is to design distributed consensus protocols based on only the local information of each agent and its neighbors to ensure that the states of the agents reach an agreement. 
In most of the previous works on consensus, e.g., \cite{ren2008consensus,ren2007high-order, li2010consensus,zhang2011optimal,seo2009consensus,yu2010second},
which address the consensus problem of
integrator-type, general linear and Lipschitz nonlinear
multi-agent systems,
the design of the consensus protocols requires the knowledge of certain connectivity of the communication graph.
The connectivity for undirected graphs can be
measured by the smallest nonzero eigenvalue of the corresponding Laplacian matrix \cite{olfati-saber2004consensus},
while for directed graphs the smallest positive real part of
the eigenvalues of the Laplacian matrix \cite{li2010consensus} or other similar
quantities are typically used \cite{yu2010second}.
Since the connectivity measures require
computing the nonzero eigenvalues of the Laplacian matrix associated
with the entire communication graph, the consensus protocols in these works
require in essence global information of the graph
which cannot be determined in a fully distributed manner.

Fully distributed consensus protocols, nevertheless, can be developed by implementing adaptive laws to
dynamically update the coupling weights of neighboring agents, thus removing
the aforementioned requirement on the global eigenvalue information. Such
adaptive consensus protocols are proposed in \cite{li2012adaptiveauto,li2011adaptive} for linear multi-agent systems,
which depend on only local information of each agent and its neighbors.
Similar adaptive schemes are presented in \cite{su2011adaptive,yu2013distributed}
for second-order nonlinear agents.
Note that the adaptive protocols in
\cite{li2012adaptiveauto,li2011adaptive,su2011adaptive,yu2013distributed}
are applicable to only undirected communication graphs.
Due to the asymmetry of the corresponding Laplacian matrices, however,
the development of distributed adaptive consensus protocols
poses a more difficult problem with directed graphs.
By introducing monotonically increasing functions as a means to provide
additional freedom for design,
a distributed adaptive consensus protocol is constructed in \cite{li2014TAC},
for directed graphs containing a directed spanning tree in which the leader is the root node.
Another adaptive protocol is presented in \cite{lv2014Auto}, which can be modified using
the $\sigma$-modification technique
so that it is robust in the presence of bounded external disturbances.
It is worth pointing out, however, that the protocols in \cite{li2014TAC,lv2014Auto} rely on
the relative states of neighboring agents, which may not be available in general.
How to design fully distributed adaptive {\it output feedback} consensus protocols
using only local output information
appears much more challenging. Except those protocols proposed in \cite{li2015Auto} for quite special minimum-phase and relative-degree-one
agents, designing fully distributed output feedback protocols for general linear multi-agent systems
with directed graphs remains to be an open issue.


The aforementioned works are concerned with the leaderless consensus problem or distributed tracking problem for the case where
the leader is of zero control input. The distributed tracking problem in the presence of a leader having a nonzero control input is generally more difficult and has been addressed in
\cite{cao2011distributed,mei2011tracking,mei2012distributed,meng2012leader,li2012adaptiveauto,li2011trackingTAC}. In particular,
the authors in \cite{cao2011distributed} present nonsmooth controllers for first- and second-order integrators
in the absence of velocity or acceleration measurements. The controllers in \cite{cao2011distributed} incorporate discontinuous functions
commonly found in the sliding mode control literature, which are meant to cope with the leader's bounded control input.
The authors in \cite{mei2011tracking,mei2012distributed,meng2012leader} address a distributed coordinated tracking and containment control
problem, for multiple Euler-Lagrange systems with one or more dynamic leaders.
Distributed static and adaptive protocols are given in \cite{li2012adaptiveauto,li2011trackingTAC}
for general linear multi-agent systems with a leader of bounded control input.
It is worth noting that one common assumption in \cite{cao2011distributed,li2012adaptiveauto,li2011trackingTAC}
is that the subgraph among the followers is undirected. The case where this subgraph is directed
remains unsolved for general linear multi-agent systems.
The main obstacle lies in the unpleasant interrelations between the nonlinear functions used to deal with the leader's
control input and the directed subgraph among followers.

\subsection{Our Contributions}

In this paper, we address the distributed adaptive output feedback consensus protocol design problem
for general linear multi-agent systems with directed communication graphs.
In this setting, the relative states of neighboring agents are not available, but only local output information is accessible. Both the cases with
and without a leader of bounded control input are studied. Note that simply combining the techniques for the state feedback case (e.g.,
those proposed in \cite{li2014TAC,lv2014Auto}) and distributed adaptive observer-type protocols (e.g., in \cite{li2012adaptiveauto}) for undirected graphs will not yield distributed adaptive output feedback consensus protocols applicable to general directed graphs. The main reason is that
the monotonically increasing functions introduced in \cite{li2014TAC,lv2014Auto}, when used for observer-type adaptive protocols in \cite{li2012adaptiveauto}, will still depend on the relative states of neighboring agents.
This motivates us to seek for novel methods to design distributed output feedback consensus protocols.

To circumvent the difficulties alluded to above, in this paper we propose a two-step, sequential observer design (SOD) method,
which consists of designing first a local observer and next a distributed observer. Here the local observer is employed to estimate from an agent's output the agent' state, while the distributed observer operates on the local state estimates and generates the control input; neither of these observers uses the state information of the agents.
Utilizing this novel SOD architecture, two types of distributed adaptive output feedback consensus protocols are developed for
the leaderless consensus problem, which exchange the local estimates among neighboring agents via the communication graph and implement adaptive laws
to update the time-varying coupling weights among the agents.
As such, these two adaptive protocols uses only the local output information and achieve leaderless consensus in a fully distributed manner,
for any strongly connected directed graph. This constitutes one of our main contributions in this paper.

Another main result of this paper concerns the leader-follower consensus with a leader of bounded control input.
We propose a distributed discontinuous adaptive output feedback protocol, to solve the consensus problem, which includes
discontinuous nonlinear functions to deal with the effect of the leaders's nonzero control input.
It is shown that the discontinuous adaptive protocols can achieve leader-follower consensus
for any directed graph containing a directed spanning tree with the leader as the root.
To attenuate the chattering phenomenon resulted from discontinuity,
distributed continuous adaptive output feedback protocols are further developed, which can ensure
the ultimate boundedness of the consensus error and the adaptive gains.
The upper bound of the consensus error is explicit derived, which can be made satisfactorily
small by appropriately tuning the design parameters.
Unlike the protocols in the previous works \cite{cao2011distributed,li2012adaptiveauto,li2011trackingTAC},
the adaptive protocols proposed herein appear to be the first available for linear multi-agent systems with general directed graphs.


\subsection{Outline of This Paper}

The rest of this paper is organized as follows. The mathematic preliminaries used in this paper are summarized in Section \ref{s2}.
Distributed adaptive output feedback consensus protocols for general linear systems with strongly connected graphs are presented in Section \ref{s3}. The leader-follower consensus problem for the case with a leader of bounded control input
is studied in Section \ref{s4}.
Numerical simulation results are presented in Section \ref{s5}. Section \ref{s6} concludes our paper.

\section{Mathematical Preliminaries}\label{s2}

{\it Notations}: Throughout this paper, the symbol
$\textbf{1}$ denotes a column vector with all entries equal to 1.
For any square matrix $A$, $\lambda_{\min}(A)$ and $\lambda_{\max}(A)$ represent the minimal and maximal eigenvalues of $A$,
respectively.
A matrix $A =[a_{ij} ]\in \mathbf{R}^{n\times n}$ is
called a nonsingular M-matrix, if $a_{ij} < 0$, $ \forall i \neq j$,
and all eigenvalues of $A$ have positive real parts.

The communication graph among the agents is specified by a directed graph $\mathcal{G}=(\mathcal{V},\mathcal{E})$, where $\mathcal{V}=\{1,\cdots,N\}$ is the set of nodes (each node represents an agent)
and $\mathcal{E}\subseteq \mathcal{V}\times\mathcal{V}$ denotes the set of edges (each edge represents a communication link between
two distinct agents).
An edge $(i,j)\in \mathcal{E}$ represents that node $i$ is a neighbor of node $j$ and node $j$ can have access to the state or output
of node $i$. A directed path from node $i_1$ to node $i_l$ is a sequence of ordered edges in the form
of $(i_k,i_{k+1})$, $k=1,\cdots,l-1$. A directed graph contains a directed spanning tree if there exists a node called the root such that the node has directed paths to all other nodes in the graph. A directed graph is strongly connected if there exists a directed
path between every pair of distinct nodes. A directed graph has a directed spanning tree if it is strongly connected, but not vice versa.

The adjacency matrix associated with the communication graph $\mathcal{G}$, denoted by $\mathcal{A}=[a_{ij}]\in \mathbf{R}^{N\times N}$, is defined as $a_{ii}=0$, $a_{ij}=1$ if $(i,j)\in\mathcal{E}$ and 0 otherwise.
The Laplacian matrix $\mathcal{L}=[l_{ij}]\in \mathbf{R}^{N\times N}$ is defined such that $l_{ii}=\sum_{j=1}^{N}a_{ij}$ and $l_{ij}=-a_{ij}$, $i\neq j$. We summarize below key facts relevant to our subsequent developments.


\begin{lemma}[\cite{ren2005consensus}]\label{lem1}
Zero is an eigenvalue of $\mathcal {L}$ with $\mathbf{1}$ as a
right eigenvector and all nonzero eigenvalues have positive real
parts. Besides, zero is a simple eigenvalue of $\mathcal {L}$ if
and only if $\mathcal {G}$ has a directed spanning tree.
\end{lemma}

\begin{lemma}[\cite{mei2014consensus}]\label{laplace}
Suppose that $\mathcal{G}$ is strongly connected. There exists a vector $r=[r_1,\cdots,r_N]^T$ with $r_i>0,i=1,\cdots,N$, such that $r^T\mathcal{L}=0$. Let $R=\mathrm{diag}(r_1,\cdots,r_N)$ \cite{yu2010second}. Then, $\hat{\mathcal{L}}\triangleq R\mathcal{L}+\mathcal{L}^TR$ is the symmetric Laplacian matrix associated with an undirected connected graph. Moreover, $\min_{\xi^Tx=0,x\neq0}\frac{x^T\hat{\mathcal{L}}x}{x^Tx}>\frac{\lambda_2(\hat{\mathcal{L}})}{N}$, where $\lambda_2(\hat{\mathcal{L}})$ denotes the smallest nonzero eigenvalue of $\hat{\mathcal{L}}$ and $\xi$ is any vector with positive entries.
\end{lemma}

\begin{lemma}[\cite{zhang2012lyapunov,li2014TAC}]\label{Mmatrix}
Consider a nonsingular M-matrix $L$. There exists a diagonal matrix $G\equiv \text{diag}(g_1,\cdots,g_N)>0$ such that $GL+L^TG>0$.
\end{lemma}

\begin{lemma}[\cite{bernstein2009matrix}]\label{ineq}
If $a$ and $b$ are nonnegative real numbers and $p$ and $q$ are positive real numbers such that $\frac{1}{p} + \frac{1}{q} = 1$, then
$ab\leq \frac{a^p}{p} + \frac{b^q}{q}$, where the equality holds if and only if $a^p = b^q$.
\end{lemma}

\begin{lemma}[\cite{slotine1991applied}]\label{comparison}
If a real function $W(t)$ satisfies the inequality $W(t) \leq -aW(t) +b $,
where $a$ and $b$ are positive constant numbers. Then,
$$W(t)\leq [W(0)-\frac{a}{b}] e^{-at} +\frac{a}{b}.$$
\end{lemma}

\section{Distributed Output Feedback Adaptive Protocols
for Strongly Connected Graphs}\label{s3}

Consider a group of $N$ identical linear dynamical systems, with the dynamics of the $i$-th agent described by
\begin{equation}\label{model1}
\begin{aligned}
&\dot{x}_{i}=Ax_{i}+Bu_{i},\\
&y_i=Cx_i,
\quad  i=1,\cdots,N,
\end{aligned}
\end{equation}
where $x_i\in\mathbf{R}^n$ is the state vector, $y_i\in\mathbf{R}^m$ is the measured output vector,
$u_i\in\mathbf{R}^{p}$ is the control input vector of the $i$-th agent, respectively, and
$A$, $B$ and $C$ are constant known matrices with compatible dimensions.

The communication graph among the $N$ agents are represented by a directed graph $\mathcal{G}$.
We assume that $\mathcal{G}$ satisfies

\begin{assumption}\label{assp1}
The communication graph $\mathcal{G}$ is strongly connected.
\end{assumption}

With the agent dynamics given in \dref{model1}, our purpose in this section is to design fully distributed output feedback consensus protocols
to solve the consensus problem, wherein by consensus, we mean that
$\lim_{t\rightarrow \infty}\|x_i(t)- x_j(t)\|=0$,
$\forall\,i,j=1,\cdots,N.$

%
%
%
%
%

We assume that the agents have access to their own outputs, i.e.,
the agents are introspective as termed in \cite{yang2011output,peymani2014h}.
The distributed adaptive output feedback protocol is proposed for each agent as follows:
\begin{equation}\label{controller1}
\begin{aligned}
\dot{v}_i&=Av_i+Bu_i+F(Cv_i-y_i),\\
\dot{w}_i&=Aw_i+Bu_i+(d_i+\rho_i)FC(\psi_i-\eta_i)+F(Cv_i-y_i),\\
u_i&=Kw_i,\\
\dot{d}_i&=(\psi_i-\eta_i)^TC^TC(\psi_i-\eta_i),\quad i=1,\cdots,N,
\end{aligned}
\end{equation}
where $\psi_i\triangleq\sum_{j=1}^{N}a_{ij}(w_{i}-w_{j})$, $\eta_i\triangleq\sum_{j=1}^{N}a_{ij}(v_{i}-v_{j})$,
$v_i\in\mathbf{R}^n$ and $w_i\in\mathbf{R}^n$ are the internal states of the protocol,
$d_i$ denotes the time-varying coupling weight associated with the $i$-th agent with $d_i(0)>0$, $K$ and $F$
are the feedback gain matrices, and $\rho_i(\cdot)$
are smooth and monotonically increasing functions
which satisfy the condition
$\rho_i(s)\geq0$ for $s>0$.

Note that the term $C(\psi_i-\eta_i)$ in \dref{controller1} implies that
the agents need to transmit the virtual outputs $Cw_i$ and $Cv_i$ of
the internal states $w_i$ and $v_i$ of their corresponding protocols to their neighbors
via the communication network $\mathcal {G}$.
In \dref{controller1}, the parameters $K$, $F$, and $\rho_i(\cdot)$
need to be determined.

Let $\xi_i\triangleq\sum_{j=1}^{N}a_{ij}(x_{i}-x_{j}),i=1,\cdots,N$, $\xi\triangleq[\xi_{1}^{T},\cdots,\xi_{N}^{T}]^{T}$, $\eta\triangleq[\eta_{1}^{T},\cdots,\eta_{N}^{T}]^{T}$, $\psi\triangleq[\psi_{1}^{T},\cdots,\psi_{N}^{T}]^{T}$, $x\triangleq[x_{1}^{T},\cdots,x_{N}^{T}]^{T}$, $v\triangleq[v_{1}^{T},\cdots,v_{N}^{T}]^{T}$, and
$w\triangleq[w_{1}^{T},\cdots,w_{N}^{T}]^{T}$.
Then, we have
\begin{equation}\label{error1}
\begin{aligned}
\xi&=(\mathcal{L}\otimes I_n)x,\\
\eta&=(\mathcal{L}\otimes I_n)v,\\
\psi&=(\mathcal{L}\otimes I_n)w.
\end{aligned}
\end{equation}
Under Assumption \ref{assp1}, it is known by Lemma \ref{lem1} that $\mathcal {L}$ has an eigenvalue at the origin with $\mathbf{1}$ being the corresponding eigenvector and all the nonzero eigenvalues of $\mathcal {L}$ have positive real parts.
By the first equality in \dref{error1}, it is easy to see that
the consensus problem is solved if and only if $\xi$ asymptotically
converges to zero. We will refer to $\xi$ as the consensus error hereafter.

By substituting \dref{controller1} into \dref{model1}, we can write the closed-loop dynamics of the
network in a compact form as
\begin{equation}\label{derror1}
\begin{aligned}
\dot{\xi}&=(I_N\otimes A)\xi+(I_N\otimes BK)\psi,\\
\dot{\eta}&=(I_N\otimes A)\eta+(I_N\otimes BK)\psi+(I_N\otimes FC)(\eta-\xi),\\
\dot{\psi}&=[I_N\otimes (A+ BK)]\psi+[\mathcal{L}(D+\rho)\otimes FC](\psi-\eta)\\
&\quad+(I_N\otimes FC)(\eta-\xi),\\
\dot{d}_i&=(\psi_i-\eta_i)^TC^TC(\psi_i-\eta_i),\quad i=1,\cdots,N,
\end{aligned}
\end{equation}
where $D\triangleq \mathrm{diag}(d_1,\cdots,d_N)$ and $\rho\triangleq \mathrm{diag}(\rho_1,\cdots,\rho_N)$.
Let $\zeta\triangleq \eta-\xi$ and $\varrho\triangleq[\varrho_1^T,\cdots,\varrho_N^T]^T=\psi-\eta$,
where $\varrho_i=\psi_i-\eta_i$, $i=1,\cdots,N$.
Then, it follows from the first two equalities in \dref{derror1} that
\begin{equation}\label{derror11}
\begin{aligned}
\dot{\zeta}&=[I_N\otimes (A+FC)]\zeta.
\end{aligned}
\end{equation}
From the last three equalities in \dref{derror1}, we find that
\begin{equation}\label{derror12}
\begin{aligned}
\dot{\varrho}&=[I_N\otimes A+\mathcal{L}(D+\rho)\otimes FC]\varrho,\\
\dot{d}_i&=\varrho_i^TC^TC\varrho_i,\quad i=1,\cdots,N.
\end{aligned}
\end{equation}
The third equality in \dref{derror1} can be rewritten as
\begin{equation}\label{derror13}
\begin{aligned}
\dot{\psi}&=[I_N\otimes (A+ BK)]\psi+[\mathcal{L}(D+\rho)\otimes FC]\varrho\\
&\quad+(I_N\otimes FC)\zeta.
\end{aligned}
\end{equation}

The following theorem designs the adaptive output feedback protocol \dref{controller1}.

\begin{theorem}\label{leaderless1}
Suppose that the communication graph $\mathcal{G}$ satisfies Assumption 1.
Then, the leaderless consensus problem of the agents in \dref{model1} can be solved under the adaptive output feedback protocol \dref{controller1}, if $A+BK$ is Hurwitz, $F=-S^{-1}C^T$, and $\rho_i=\varrho_i^TS\varrho_i$, where $S>0$ is a
solution to the linear matrix inequality (LMI)
\begin{equation}\label{lmic}
A^TS+SA-2C^TC<0.
\end{equation}
Moreover, each coupling weight $d_{i}$ converges to some finite steady-state value.
\end{theorem}

\begin{proof}
Since $F=-S^{-1}C^T$, it follows from \dref{lmic} that
$$(A+FC)^TS+S(A+FC)=A^TS+SA-2C^TC<0.$$
Thus, $A+FC$ is Hurwitz and $\zeta$ in \dref{derror11} asymptotically converges to zero.

Next, we show that $\varrho$ in \dref{derror12} converges to zero.
To this end, consider the Lyapunov function candidate
\begin{equation}\label{lya1}
\begin{aligned}
V_{1}=&\frac{1}{2}\sum_{i=1}^{N}r_i(2d_i+\rho_i)\varrho_i^TS\varrho_i
+\frac{1}{2}\sum_{i=1}^{N}r_i\tilde{d}_i^2,
\end{aligned}
\end{equation}
where $r=[r_1,\cdots,r_N]^T$ is the
left eigenvector of $\mathcal{L}$ associated with the zero eigenvalue,
$\tilde{d}_i\triangleq d_{i}-\alpha$, with $\alpha$
being a positive constant to be determined subsequentially.
Since Assumption \ref{assp1} holds, it follows from Lemma \ref{lem1} that $R\triangleq \mathrm{diag}(r_1,\cdots,r_N)>0$.
Since $d_i(0) > 0$ and $\dot{d}_i\geq0$, it follows that
$d_i(t)>0$.
Noting further that $\rho(t)\geq0$,
it is not difficult to see that $V_1$ is positive definite
with respect to $\varrho$ and $\tilde{d}$.


The time derivative of $V_{1}$ along the trajectory of \dref{derror12} is given by
\begin{equation}\label{dotlya1}
\begin{aligned}
\dot{V}_1&=\sum_{i=1}^{N}[2r_i(d_i+\rho_i)\varrho_i^TS\dot{\varrho}_i+r_i\rho_i\dot{d}_i]+\sum_{i=1}^{N}r_i\tilde{d}_i\dot{\tilde{d}}_i\\
&=\varrho^T[(D+\rho)R\otimes (SA+A^TS)\\
&\quad-(D+\rho)\hat{\mathcal{L}}(D+\rho)\otimes C^TC\\
&\quad+(\rho R+DR-\alpha R)\otimes C^TC]\varrho,
\end{aligned}
\end{equation}
where $\hat{\mathcal{L}}\triangleq R\mathcal{L}+\mathcal{L}^TR$.
Let $\overline{\varrho}=((D+\rho)\otimes I_n)\varrho$. Then, we have
\begin{equation*}
\begin{aligned}
\overline{\varrho}^T((D+\rho)^{-1}r\otimes \textbf{1})&=\varrho^T(r\otimes \textbf{1})\\
&=(w-v)^T(\mathcal{L}^Tr\otimes \textbf{1})=0,
\end{aligned}
\end{equation*}
where we have used the fact that $r^T\mathcal{L}=0$. Since every entry of $r$ is positive, it is obvious that every entry of $(D+\rho)^{-1}r\otimes \textbf{1}$ is also positive. In light of Lemma \ref{laplace}, we get
\begin{equation}\label{lapineq}
\begin{aligned}
\overline{\varrho}^T(\hat{\mathcal{L}}\otimes I_n)\overline{\varrho} &>\frac{\lambda_2(\hat{\mathcal{L}})}{N}\overline{\varrho}^T\overline{\varrho}\\
&=\frac{\lambda_2(\hat{\mathcal{L}})}{N}\varrho^T[(D+\rho)^2\otimes I_n]\varrho.
\end{aligned}
\end{equation}
Note that
\begin{equation}\label{drhoineq}
\begin{aligned}
\varrho^T[(D+\rho) R\otimes C^TC]\varrho\leq&\frac{\lambda_2(\hat{\mathcal{L}})}{4N}\varrho^T[(D+\rho)^2\otimes C^TC]\varrho\\
&+\frac{N}{\lambda_2(\hat{\mathcal{L}})}\varrho^T(R^2\otimes C^TC)\varrho,
\end{aligned}
\end{equation}
and
\begin{equation}\label{drhoalineq}
\begin{aligned}
&-\frac{\lambda_2(\hat{L})}{4N}\varrho^T[(D+\rho)^2\otimes C^TC]\varrho\\
&-\varrho^T[(\alpha R-\frac{N}{\lambda_2(\hat{\mathcal{L}})}R^2)\otimes C^TC]\varrho\\
&\qquad\qquad\leq-2\varrho^T[(D+\rho)R\otimes C^TC]\varrho,
\end{aligned}
\end{equation}
where we have used Lemma \ref{ineq} and chosen $\alpha\geq \frac{5N\lambda_{\max}(R)}{\lambda_2(\hat{\mathcal{L}})}$ to get the inequalities.
By substituting \dref{lapineq}, \dref{drhoineq}, and \dref{drhoalineq} into \dref{dotlya1}, we then obtain
\begin{equation}\label{psieta}
\begin{aligned}
\dot{V}_1&\leq\varrho^T[(D+\rho)R\otimes (SA+A^TS)]\varrho\\
&\quad-\frac{3\lambda_2(\hat{\mathcal{L}})}{4N}\varrho^T((D+\rho)^2\otimes C^TC)\varrho\\
&\quad-\varrho^T[(\alpha R-\frac{N}{\lambda_2(\hat{\mathcal{L}})}R^2)\otimes C^TC]\varrho\\
&\leq\varrho^T[(D+\rho)R\otimes (A^TS+SA-2C^TC)]\varrho\\
&\leq\varrho^T[D(0)R\otimes (A^TS+SA-2C^TC)]\varrho\\
&\leq0,
\end{aligned}
\end{equation}
where we have used the facts that $D(t)\geq D(0)$ and $\rho\geq0$ to arrive at the third inequality
and used \dref{lmic} to obtain the last inequality.
Consequently, $V_1(t)$ is bounded and so is each $d_i$. Noting that $\dot{d}_i\geq 0$, we assert that each coupling weight $d_i$ converges to some finite value. Furthermore, $\dot{V}_1\equiv0$
 implies that $\varrho\equiv0$.
By LaSalle's Invariance principle \cite{khalil2002nonlinear},
it follows that $\varrho$ asymptotically converges to zero.
In light of this, together with the fact that
$\varrho$ asymptotically converges to zero
and $D(t)$ is bounded, it is easy to see that $[\mathcal{L}(D+\rho)\otimes FC]\varrho$ asymptotically converges to zero.
Since $A+BK$ is Hurwitz and $\zeta$ converge asymptotically to zero, it follows that $\psi$ also asymptotically converges to zero.
In conclusion, we have shown that all $\zeta$, $\varrho$, and $\psi$ asymptotically converge to zero, which, in virtue of
the definitions of $\zeta$, $\varrho$, and $\psi$, implies that the consensus error $\xi$
converges asymptotically to zero. Thus, the agents achieve consensus. \hfill $\blacksquare$
\end{proof}

\begin{remark}
As shown in \cite{li2010consensus}, a necessary and sufficient condition for the existence of a solution $S>0$ to the LMI \dref{lmic} is that
$(A,C)$ is detectable. Therefore, an sufficient existence condition of an adaptive protocol \dref{controller1} satisfying Theorem \ref{leaderless1} is that $(A,B,C)$ is stabilizable and detectable. Note that the adaptive  protocol \dref{controller1} can be designed
by solving the algebraic Riccati equation: $\bar{P} A^T+A\bar{P}+I-\bar{P}C^TC\bar{P}=0$,
as in \cite{tuna2008lqr,zhang2011optimal}.
In this case, the parameters in \dref{controller1} can be chosen
as $F=-\bar{P} C^T$ and $\rho_i=\varrho_i^T\bar{P}^{-1}\varrho_i$.
Evidently, the adaptive protocol \dref{controller1} uses only the agent dynamics and the local output information,
and therefore is fully distributed.
\end{remark}

\begin{remark}
In contrast to the adaptive output feedback protocols in \cite{li2012adaptiveauto}, which are applicable to only undirected graphs, the adaptive protocol \dref{controller1} works for directed graphs,
provided that the graphs are strongly connected.
That this protocol is both fully distributed and only dependent on local output information is made possible by a novel
sequential observer design (SOD) architecture, which consists of a local observer and a graph-based distributed observer:
the local observer (the first equation in \dref{controller1}) estimates the state of each agent,
while the distributed observer (the second equation in \dref{controller1})
provides feedback based on estimated relative states, thus ensuring that $\zeta$, $\varrho$, and $\psi$ converge to zero and subsequently
the consensus error $\xi$ converges to zero.
In essence, the SOD method reduces the closed-loop network dynamics \dref{derror1} into a upper-triangular form,
where the first two dynamics \dref{derror11} and \dref{derror12} are independent and
\dref{derror13} relies on \dref{derror11} and \dref{derror12}. 
\end{remark}

Apart from the protocol \dref{controller1},
an alternative adaptive output feedback protocol can be constructed as follows:
\begin{equation}\label{controller2}
\begin{aligned}
&\dot{v}_i=Av_i+Bu_i+F(Cv_i-y_i),\\
&\dot{w}_i=Aw_i+Bu_i+(d_i+\rho_i)BK(\psi_i-\eta_i)+F(Cv_i-y_i),\\
&u_i=Kw_i,\\
&\dot{d}_i=(\psi_i-\eta_i)^T\Omega(\psi_i-\eta_i),\quad i=1,\cdots,N,
\end{aligned}
\end{equation}
where $\Omega\in \mathbf{R}^{n\times n}$ is a feedback gain matrix,
and the rest of the variables are defined as in \dref{controller1}.
The parameters $K$, $F$, $\Omega$, and $\rho_i$ in \dref{controller2}
need to be determined. The difference between these two protocols will be elaborated subsequently.

\begin{theorem}\label{leaderless2}
Suppose that Assumption 1 holds.
The consensus problem of the $N$ agents described by \dref{model1}
is solved by the adaptive output feedback protocol \dref{controller2},
whenever $F$ satisfying that $A+FC$ is Hurwitz, $K=-B^TP^{-1}$,
$\Omega=P^{-1}BB^TP^{-1}$, and $\rho_i=\varrho_i^TP^{-1}\varrho_i$
where $\varrho_i$ is defined as in \dref{derror12} and
$P>0$ is a solution to the LMI:
\begin{equation}\label{lmib}
PA^T+AP-2BB^T<0,
\end{equation}
In addition, the coupling weight $d_{i}$ converge to finite steady-state values.
\end{theorem}

\begin{proof}
From \dref{controller2} and \dref{model1}, we can obtain the closed-loop network dynamics as follows:
\begin{equation}\label{derror200}
\begin{aligned}
\dot{\zeta}&=(I_N\otimes (A+FC))\zeta,\\
\dot{\varrho}&=[I_N\otimes A+\mathcal{L}(D+\rho)\otimes BK]\varrho,\\
\dot{\psi}&=[I_N\otimes (A+ BK)]\psi+[\mathcal{L}(D+\rho)\otimes BK]\varrho\\
&\quad+(I_N\otimes FC)\zeta,\\
\dot{d}_i&=\varrho_i^T\Omega\varrho_i,\quad i=1,\cdots,N,
\end{aligned}
\end{equation}
where the variables $\zeta$, $\varrho$, and $\psi$ are defined as in \dref{derror11}, \dref{derror12},
and \dref{derror13}, respectively.
The convergence of $\zeta$ in \dref{derror200} to zero is obvious. The convergence of $\varrho$ in \dref{derror200}
can be proved similarly as in the proof of Theorem 1, by using the Lyapunov function candidate
$$\begin{aligned}
V_{2}=&\frac{1}{2}\sum_{i=1}^{N}r_i(2d_i+\rho_i)\varrho_i^TP^{-1}\varrho_i
+\frac{1}{2}\sum_{i=1}^{N}r_i\tilde{d}_i^2.
\end{aligned}
$$
The rest of the proof can be completed analogously as in the proof of Theorem 1. We omit the details
for brevity.\hfill $\blacksquare$
\end{proof}

\begin{remark}
The LMI \dref{lmib} is dual to \dref{lmic}. Therefore, a sufficient condition for the existence of \dref{controller2} satisfying Theorem 2 is also
that $(A,B,C)$ is stabilizable and detectable.
Note that the adaptive protocol \dref{controller1}
transmits $Cv_i$ and $Cw_i$ between neighboring agents while the adaptive protocol \dref{controller2}
transmits $v_i$ and $w_i$. The sum of the dimensions of $Cv_i$ and $Cw_i$ is generally
lower than that of $v_i$ and $w_i$.
Therefore, the adaptive protocol \dref{controller1} is more favorable,
because of a lower communication burden.
\end{remark}

\section{Adaptive output feedback protocols for Leader-follower Graphs}\label{s4}
Theorem \ref{leaderless1} and Theorem \ref{leaderless2} in the previous section are applicable to strongly connected directed graphs.
In this section, we extend our analysis to leader-follower consensus problems,
alternatively known as distributed tracking.

Consider a group of $N+1$ agents with general linear dynamics described by \dref{model1},
indexed by $0,\cdots,N$.
Suppose that the agents indexed by $1,\cdots,N$, are the $N$ followers and the agent indexed by 0 is the leader.
Under many circumstances, the leader may need to implement control action to regulate the final consensus trajectory.
We assume that in general the leader's control input is bounded, i.e.,
the following assumption holds.
\begin{assumption}\label{omega}
There exists a positive constant $\omega$ such that $\|u_0\|\leq\omega$.
\end{assumption}
Moreover, we assume that the communication graph $\widehat{\mathcal{G}}$ among the $N+1$ agents satisfies
\begin{assumption}\label{assp2}
The graph $\widehat{\mathcal{G}}$ contains a directed spanning tree with the leader as the root node.
\end{assumption}

Under Assumption \ref{assp2}, the Laplacian matrix associated with $\widehat{\mathcal{G}}$ can be partitioned as $\widehat{\mathcal{L}}=
\begin{bmatrix}
0& 0_{1\times N}\\
\widehat{\mathcal{L}}_2&\widehat{\mathcal{L}}_1
\end{bmatrix}$,
where $\widehat{\mathcal {L}}_2\in\mathbf{R}^{N\times 1}$ and $\widehat{\mathcal
{L}}_1\in\mathbf{R}^{N\times N}$.
It is easy to verify that $\widehat{\mathcal
{L}}_1$ is a nonsingular M-matrix.


In this section we propose new distributed output feedback consensus protocols
that solve the leader-follower consensus problem, in the sense that $\lim_{t\rightarrow \infty}\|x_i(t)- x_0(t)\|=0$,
$\forall\,i=1,\cdots,N.$

\subsection{Discontinuous Adaptive Consensus Protocols}

Based on the relative estimates of the states of neighboring agents, the following distributed discontinuous adaptive controller is proposed for each follower:
\begin{equation}\label{controller3}
\begin{aligned}
\dot{\tilde{v}}_i=&A\tilde{v}_i+Bu_i+F(C\tilde{v}_i-y_i),\\
\dot{\tilde{w}}_i=&A\tilde{w}_i+B[u_i-\beta h(B^TS(\tilde{\psi}_i-\tilde{\eta}_i))]\\&+(d_i+\rho_i)FC(\tilde{\psi}_i-\tilde{\eta}_i)+F(C\tilde{v}_i-y_i),\\
u_i=&K[\tilde{w}_i-\beta h(B^TQ\tilde{\eta}_i)],\\
\dot{d}_i=&(\tilde{\psi}_i-\tilde{\eta}_i)^TC^TC(\tilde{\psi}_i-\tilde{\eta}_i),\quad i=1,\cdots,N,
\end{aligned}
\end{equation}
where 
$\dot{\tilde{v}}_0=A\tilde{v}_0+Bu_0+F(C\tilde{v}_0-y_0)$, $w_0=0$, $\tilde{\psi}_i\triangleq\sum_{j=0}^{N}a_{ij}(\tilde{w}_{i}-\tilde{w}_{j})$, $\tilde{\eta}_i\triangleq\sum_{j=0}^{N}a_{ij}(\tilde{v}_{i}-\tilde{v}_{j})$, $\rho_i=(\tilde{\psi}_i-\tilde{\eta}_i)^TS(\tilde{\psi}_i-\tilde{\eta}_i)$,
$S>0$ is a solution to the LMI \dref{lmic}, $Q>0$, $\beta$ is a positive constant, $d_i$ denotes the time-varying coupling weight associated with the $i$-th follower with $d_i(0)\geq0$, and
the
nonlinear function $h(\cdot)$ is defined such that
for $z\in\mathbf{R}^n$,
\begin{equation}\label{satu2}
h(z)=\begin{cases}\frac{z}{\|z\|} & \text{if}~\|z\|\neq 0,\\
0 & \text{if}~\|z\|=0.
\end{cases}
\end{equation}
In \dref{controller3}, the parameters $K$, $F$, $Q$, and $\beta$ are to be determined.

Let $\tilde{\xi}_i\triangleq\sum_{j=0}^{N}a_{ij}(x_{i}-x_{j}),i=1,\cdots,N$, $\tilde{\xi}\triangleq[\tilde{\xi}_{1}^{T},\cdots,\tilde{\xi}_{N}^{T}]^{T}$, $\tilde{\eta}\triangleq[\tilde{\eta}_{1}^{T},\cdots,\tilde{\eta}_{N}^{T}]^{T}$, $\tilde{\psi}\triangleq[\tilde{\psi}_{1}^{T},\cdots,\tilde{\psi}_{N}^{T}]^{T}$, and $x\triangleq[x_{1}^{T},\cdots,x_{N}^{T}]^{T}$, $\tilde{v}\triangleq[\tilde{v}_{1}^{T},\cdots,\tilde{v}_{N}^{T}]^{T}$, $\tilde{w}\triangleq[\tilde{w}_{1}^{T},\cdots,\tilde{w}_{N}^{T}]^{T}$. Then, we have
\begin{equation}\label{xi}
\begin{aligned}
\tilde{\xi}&=(\widehat{\mathcal{L}}_1\otimes I_n)(x-\textbf{1}\otimes x_0),\\
\tilde{\eta}&=(\widehat{\mathcal{L}}_1\otimes I_n)(\tilde{v}-\textbf{1}\otimes \tilde{v}_0),\\
\tilde{\psi}&=(\widehat{\mathcal{L}}_1\otimes I_n)\tilde{w}.
\end{aligned}
\end{equation}
By the first equality in \dref{xi}, it is easy to see that the leader-follower consensus problem is solved if and only if the consensus error $\tilde{\xi}$ asymptotically converges to zero.

By substituting protocol \dref{controller3} into \dref{model1}, we can get the closed-loop dynamics of the network as follows:
$$\begin{aligned}
\dot{\tilde{\xi}}&=(I_N\otimes A)\tilde{\xi}+(I_N\otimes BK)\tilde{\psi}\\&\quad-(\widehat{\mathcal{L}}_1\otimes B)[\beta M(\tilde{\eta})+\textbf{1}\otimes u_0],\\
\dot{\tilde{\eta}}&=(I_N\otimes A)\tilde{\eta}+(I_N\otimes BK)\tilde{\psi}+(I_N\otimes FC)(\tilde{\eta}-\tilde{\xi})\\&\quad-(\widehat{\mathcal{L}}_1\otimes B)[\beta M(\tilde{\eta})+\textbf{1}\otimes u_0],\\
\dot{\tilde{\psi}}&=[I_N\otimes (A+BK)]\tilde{\psi}+[\widehat{\mathcal{L}}_1(D+\rho)\otimes FC](\tilde{\psi}-\tilde{\eta})\\&\quad-(\widehat{\mathcal{L}}_1\otimes B)\beta [M(\tilde{\eta})+H(\tilde{\psi}-\tilde{\eta})]\\
&\quad+(I_N\otimes FC)(\tilde{\eta}-\tilde{\xi})+(\widehat{\mathcal{L}}_1\otimes FC)[\textbf{1}\otimes(v_0-x_0)],\\
\dot{d}_i&=(\tilde{\psi}_i-\tilde{\eta}_i)^TC^TC(\tilde{\psi}_i-\tilde{\eta}_i),\quad i=1,\cdots,N,
\end{aligned}$$
where $H(\tilde{\psi}-\tilde{\eta})\triangleq[h(B^TS(\tilde{\psi}_1-\tilde{\eta}_1))^T,\cdots,h(B^TS(\tilde{\psi}_N-\tilde{\eta}_N))^T]^T$, $M(\tilde{\eta})\triangleq[h(B^TQ\eta_1)^T,\cdots,h(B^TQ\eta_N)^T]^T$,
and $D$, $\rho$ are defined as in \dref{derror1}, .
Let $\tilde{\zeta}\triangleq \tilde{\eta}-\tilde{\xi}$ and $\tilde{\varrho}\triangleq[\tilde{\varrho}_1^T,\cdots,\tilde{\varrho}_N^T]^T=\tilde{\psi}-\tilde{\eta}$.
Then, we obtain 
\begin{equation}\label{derror2}
\begin{aligned}
\dot{\tilde{\zeta}}=&[I_N\otimes (A+FC)]\tilde{\zeta},\\
\dot{\tilde{\varrho}}=&[I_N\otimes A+\widehat{\mathcal{L}}_1(D+\rho)\otimes FC]\tilde{\varrho}\\
&-(\widehat{\mathcal{L}}_1\otimes B)(\beta H(\tilde{\varrho})-\textbf{1}\otimes u_0)\\
&+(\widehat{\mathcal{L}}_1\otimes FC)(\textbf{1}\otimes e_0),\\
\dot{\eta}=&[I_N\otimes (A+BK)]\tilde{\eta}+(I_N\otimes BK)\tilde{\varrho}\\
&+(I_N\otimes FC)\tilde{\zeta}-(\widehat{\mathcal{L}}_1\otimes B)[\beta M(\tilde{\eta})+\textbf{1}\otimes u_0],\\
\dot{d}_i=&\tilde{\varrho}_i^TC^TC\tilde{\varrho}_i,\quad i=1,\cdots,N,
\end{aligned}
\end{equation}
where $e_0=\tilde{v}_0-x_0$ is defined as the state estimation error of the leader.

The following result provides a sufficient condition ensuring that the adaptive protocol \dref{controller3}
achieves leader-follower consensus.

\begin{theorem}\label{thmlf1}
Suppose that Assumptions \ref{omega} and \ref{assp2} holds.
Then, the leader-follower consensus problem of the agents in \dref{model1} can be solved under the adaptive output feedback protocol \dref{controller3} where $K$, $F$, and $\rho_i$ are designed as in Theorem \ref{leaderless1}, $\beta$ is chosen such that $\beta\geq\omega$,
and $Q>0$ satisfies
\begin{equation}\label{lmic2}
Q(A+BK)+(A+BK)^TQ>0.
\end{equation}
In addition, each coupling weight $d_i$ converges to some finite steady-state value.
\end{theorem}

\begin{proof}
The convergence of $\zeta$ in \dref{derror2} to zero is obvious. To show the convergence of $\varrho$ in \dref{derror2}, we construct the
Lyapunov function candidate
\begin{equation}\label{lya3}
\begin{aligned}
V_{3}=&\frac{1}{2}\sum_{i=1}^Ng_i(2d_i+\rho_i)\tilde{\varrho}_i^TS\tilde{\varrho}_i+\frac{1}{2}\sum_{i=1}^Ng_i\tilde{d}_i^2+\gamma e_0^TSe_0,
\end{aligned}
\end{equation}
where $G\triangleq \text{diag}(g_1,\cdots,g_N)$ is a positive definite matrix such that $G\widehat{\mathcal{L}}_1+\widehat{\mathcal{L}}_1^TG>0$, $\tilde{d}_i\triangleq d_{i}-\alpha$, $\alpha$ and $\gamma$ are positive constants to be determined later. Due to the fact that $\widehat{\mathcal{L}}_1$ is a nonsingular M-matrix, the existence of such a positive definite matrix $G$ is ensured by Lemma \ref{Mmatrix}. Then, it is not difficult to see that $V_3$ is positive definite with respect to the variables $\tilde{\varrho}_i$,
$\tilde{d}_i$, and $e_0$.

The time derivative of $V_{3}$ along the trajectory of \dref{derror2} is given by
\begin{equation}\label{dotlya3}
\begin{aligned}
\dot{V}_3=&\sum_{i=1}^N[2g_i(d_i+\rho_i)\tilde{\varrho}_i^TS\dot{\tilde{\varrho}}_i+g_i\rho_i\dot{d_i}]+\sum_{i=1}^Ng_i\tilde{d}_i\dot{\tilde{d}}_i\\
&+\gamma e_0^T[S(A+FC)+(A+FC)^TS]e_0\\
\leq&\tilde{\varrho}^T[(D+\rho)G\otimes(SA+A^TS)-\lambda_0(D+\rho)^2\otimes C^TC\\
&+(\rho+D-\alpha I_N)G\otimes C^TC]\tilde{\varrho}-\gamma e_0^TWe_0\\
&-2\tilde{\varrho}^T[(D+\rho)G\widehat{\mathcal{L}}_1\otimes SB](\beta H(\tilde{\varrho})-\textbf{1}\otimes u_0)\\
&-2\tilde{\varrho}^T[(D+\rho)G\widehat{\mathcal{L}}_1\otimes C^TC](\textbf{1}\otimes e_0),
\end{aligned}
\end{equation}
where $\lambda_0$ denotes the smallest eigenvalue of $G\widehat{\mathcal{L}}_1+\widehat{\mathcal{L}}_1^TG$ and $W\triangleq-SA-A^TS+2C^TC$ is a positive definite matrix.
Using Lemma \ref{ineq}, we then find that
\begin{equation}\label{ineq1}
\begin{aligned}
\tilde{\varrho}^T[(\rho+D)G\otimes C^TC]\tilde{\varrho}
\leq&\frac{\lambda_0}{3}\tilde{\varrho}^T[(D+\rho)^2\otimes C^TC]\tilde{\varrho}\\
&+\frac{3}{4\lambda_0}\tilde{\varrho}^T[G^2\otimes C^TC]\tilde{\varrho},
\end{aligned}
\end{equation}
and
\begin{equation}\label{ineq2}
\begin{aligned}
&-2\tilde{\varrho}^T[(D+\rho)G\widehat{\mathcal{L}}_1\otimes C^TC](\textbf{1}\otimes e_0)\\
&\qquad\leq\frac{\lambda_0}{3}\tilde{\varrho}^T[(D+\rho)^2\otimes C^TC]\tilde{\varrho}\\
&\quad\qquad+\frac{3\lambda_{\max}(C^TC)\widehat{\mathcal{L}}_2^TGG\widehat{\mathcal{L}}_2}{\lambda_0\lambda_{\min}(W)}e_0^TWe_0.
\end{aligned}
\end{equation}
Note that
\begin{equation}\label{norm}
\begin{aligned}
\tilde{\varrho}_i^TSBh(B^TS\tilde{\varrho}_i)&=\|B^TS\tilde{\varrho}_i\|,\\
\tilde{\varrho}_i^TSBh(B^TS\tilde{\varrho}_j)&\leq\|B^TS\tilde{\varrho}_i\|\|h(B^TS\tilde{\varrho}_j)\|\\&=\|B^TS\tilde{\varrho}_i\|.
\end{aligned}
\end{equation}
This leads to
\begin{equation}\label{h}
\begin{aligned}
-\tilde{\varrho}^T[(D+&\rho)G\widehat{\mathcal{L}}_1\otimes SB]\beta H(\tilde{\varrho})\\
&=-\sum_{i=1}^N(d_i+\rho_i)g_i\beta\tilde{\varrho}_i^TSB[a_{i0}h(B^TS\tilde{\varrho}_i)\\
&\quad+\sum_{j=1}^Na_{ij}(h(B^TS\tilde{\varrho}_i)-h(B^TS\tilde{\varrho}_j))]\\
&\leq-\sum_{i=1}^N(d_i+\rho_i)g_i\beta a_{i0}\|B^TS\tilde{\varrho}_i\|,
\end{aligned}
\end{equation}
where we have used \dref{norm} to arrive at the inequality.
Noting that $\widehat{\mathcal{L}}_1\textbf{1}=-\widehat{\mathcal{L}}_2$, we have
\begin{equation}\label{u0}
\begin{aligned}
\tilde{\varrho}^T[(D+&\rho)G\widehat{\mathcal{L}}_1\otimes SB](\textbf{1}\otimes u_0)\\
&=-\sum_{i=1}^N(d_i+\rho_i)g_ia_{i0}\tilde{\varrho}_i^TSBu_0\\
&\leq\sum_{i=1}^N(d_i+\rho_i)g_i\omega a_{i0}\|B^TS\tilde{\varrho}_i\|.
\end{aligned}
\end{equation}
Substituting \dref{ineq1}, \dref{ineq2}, \dref{h}, and \dref{u0} into \dref{dotlya3} yields
\begin{equation}\label{dotlya32}
\begin{aligned}
\dot{V}_3\leq&\tilde{\varrho}^T[(D+\rho)G\otimes(SA+A^TS)-\frac{\lambda_0}{3}(D+\rho)^2\otimes C^TC\\
&-(\alpha G-\frac{3}{4\lambda_0}G^2)\otimes C^TC]\tilde{\varrho}-e_0^TWe_0,
\end{aligned}
\end{equation}
where we have chosen $\gamma=1+\frac{3\lambda_{\max}(C^TC)\widehat{\mathcal{L}}_2^TGG\widehat{\mathcal{L}}_2}{\lambda_0\lambda_{\min}(W)}$.
By selecting
\begin{equation}\label{alpha}
\alpha\geq\frac{15\lambda_{\max}(G)}{4\lambda_0},
\end{equation}
it follows from Lemma \ref{ineq} that
\begin{equation}\label{ineq3}
\begin{aligned}
&\tilde{\varrho}^T[(\frac{\lambda_0}{3}(D+\rho)^2+(\alpha G-\frac{3}{4\lambda_0}G^2))\otimes C^TC]\tilde{\varrho}\\
&\qquad\geq\tilde{\varrho}^T[(D+\rho)G\otimes 2C^TC]\tilde{\varrho}.
\end{aligned}
\end{equation}
Substituting \dref{ineq3} into \dref{dotlya32} gives
\begin{equation}\label{dotlya33}
\begin{aligned}
\dot{V}_3\leq&-\tilde{\varrho}^T[(D+\rho)G\otimes W]\tilde{\varrho}-e_0^TWe_0,
\end{aligned}
\end{equation}

Let $\chi=[\tilde{\varrho}^T(\sqrt{(D+\rho)G}\otimes I_n),e_0^T]^T$. Then we can get from \dref{dotlya33} that
\begin{equation}\label{dotlya34}
\begin{aligned}
\dot{V}_{3}&\leq -\chi^T(I_{N+1}\otimes W)\chi\leq 0,
\end{aligned}
\end{equation}
where the last inequality follows directly from the definition of $W$.
As a result, $V_3(t)$ is bounded, and so are $\tilde{\varrho}$, $e_0$ and $d_i$, which, by the definition of $\chi$, implies that $\chi$ is bounded. Since by Assumption \ref{omega}, $u_0$ is bounded, this implies that $\dot{\tilde{\varrho}}$ is bounded, which further implies that $\dot{\chi}$ is bounded. Since $V_3(t)$ is nonincreasing and bounded from below by zero, it has a finite limit $V_3^{\infty}$ as $t\rightarrow \infty$. Integrating the first inequality of \dref{dotlya34}, we obtain
$$\int_0^{\infty}\chi^T(I_{N+1}\otimes W)\chi dt\leq V_3(0)-V_3^{\infty}.$$
Thus, $\int_0^{\infty}\chi^T(I_{N+1}\otimes W)\chi dt$ exists and is finite. In view of the fact that both $\chi$ and $\dot{\chi}$ are bounded, it is straightforward to see that $2\chi^T(I_{N+1}\otimes W)\dot{\chi}$ is also bounded, which in turn ensures the uniform continuity of $\chi^T(I_{N+1}\otimes W)\chi$. Therefore, by Barbalat's Lemma \cite{khalil2002nonlinear}, we can establish that $\chi^T(I_{N+1}\otimes W)\chi\rightarrow 0$ as $t\rightarrow \infty$. Noting that $\chi^T(I_{N+1}\otimes W)\chi\rightarrow 0$ equals to $\chi\rightarrow 0$ and thereby $\tilde{\varrho}$ asymptotically converges to zero. Noting that $\dot{d}_i\geq 0$, the boundedness of $d_i$ implies that each coupling weight $d_i$ converges to some finite value.

Next, we show the convergence of $\tilde{\eta}$ in \dref{derror2}. Consider the following Lyapunov function candidate:
\begin{equation}\label{lya4}
\begin{aligned}
V_{4}=&\tilde{\eta}^T(I_N\otimes Q)\tilde{\eta}+\gamma_1\tilde{\zeta}^T(I_N\otimes S)\tilde{\zeta}+\gamma_2V_3,
\end{aligned}
\end{equation}
where $\gamma_1$ and $\gamma_2$ are positive constants to be determined later. Since $Q$ and $S$ are positive definite,
it is easy to see that $V_4$ is also positive definite with respect to $\tilde{\eta}$, $\tilde{\zeta}$, $\tilde{\varrho}_i$,
$\tilde{d}_i$, and $e_0$.
The time derivative of $V_{4}$ along the trajectory of \dref{derror2} is given by
\begin{equation}\label{dotlya4}
\begin{aligned}
\dot{V}_4=&\tilde{\eta}^T[I_N\otimes (Q(A+BK)+(A+BK)^TQ)]\tilde{\eta}\\
&+2\tilde{\eta}^T(I_N\otimes QBK)\tilde{\varrho}+2\tilde{\eta}^T(I_N\otimes QFC)\tilde{\zeta}\\
&-2\tilde{\eta}^T(\widehat{\mathcal{L}}_1\otimes QB)(\beta M(\tilde{\eta})+\textbf{1}\otimes u_0)+\gamma_2\dot{V}_3\\
&+\gamma_1\tilde{\zeta}^T[I_N\otimes (S(A+FC)+(A+FC)^TS)]\tilde{\zeta}.
\end{aligned}
\end{equation}

By using Lemma \ref{ineq}, we can get that
\begin{equation}\label{ineq4}
\begin{aligned}
&2\tilde{\eta}^T(I_N\otimes QBK)\tilde{\varrho}
\leq\frac{1}{4}\tilde{\eta}^T(I_N\otimes X)\tilde{\eta}+\frac{4\lambda_{\max}^2(\Gamma)}{\lambda_{\min}(X)}\tilde{\varrho}^T\tilde{\varrho},\\
&2\tilde{\eta}^T(I_N\otimes QFC)\tilde{\zeta}
\leq\frac{1}{4}\tilde{\eta}^T(I_N\otimes X)\tilde{\eta}+\frac{4\lambda_{\max}^2(QFC)}{\lambda_{\min}(X)}\tilde{\zeta}^T\tilde{\zeta},
\end{aligned}
\end{equation}
where $X=-(QA+A^TQ-2\Gamma)$ is a positive definite matrix and $\Gamma=QBB^TQ$.
Substituting \dref{dotlya33} and \dref{ineq4} into \dref{dotlya4}, we obtain that
\begin{equation}\label{dotlya41}
\begin{aligned}
\dot{V}_4\leq&-\frac{1}{2}\tilde{\eta}^T[I_N\otimes X]\tilde{\eta}\\
&-2\tilde{\eta}^T(\widehat{\mathcal{L}}_1\otimes QB)(\beta M(\tilde{\eta})+\textbf{1}\otimes u_0),
\end{aligned}
\end{equation}
where we have used the fact that $D-I>0$ and chosen
\begin{equation}\label{gamma12}
\begin{aligned}
\gamma_1 &\geq\frac{4\lambda_{\max}^2(QFC)}{\lambda_{\min}(X)\lambda_{\min}(W)},\\ \gamma_2 &\geq\frac{4\lambda_{\max}^2(\Gamma)}{\lambda_{\min}(X)\lambda_{\min}(G)\lambda_{\min}(W)}.
\end{aligned}
\end{equation}

Similarly as in \dref{h} and \dref{u0}, we can show that
$$-2\tilde{\eta}^T(\widehat{\mathcal{L}}_1\otimes QB)(\beta M(\tilde{\eta})+\textbf{1}\otimes u_0)\leq0.$$
Thus, we follows from \dref{dotlya41} that
$$\begin{aligned}
\dot{V}_4\leq&-\frac{1}{2}\tilde{\eta}^T[I_N\otimes X]\tilde{\eta}
\leq0.
\end{aligned}$$
Therefore, $V_4$ is bounded and so is $\tilde{\eta}$. By following the same steps in showing the convergence of $\tilde{\varrho}$, we are able to obtain that $\tilde{\eta}$ asymptotically converges to zero, which, in virtue of the definitions of $\tilde{\zeta}$ and $\tilde{\eta}$, further implies that the consensus error $\xi$ asymptotically converges to zero. That is, the leader-follower consensus problem is solved.
\hfill $\blacksquare$
\end{proof}

Compared to the adaptive protocol \dref{controller1} for the leaderless case in the previous section,
the adaptive protocol \dref{controller3} contains the nonlinear components $h(B^TS\tilde{\varrho}_i)$ and $h(B^TQ\tilde{\eta}_i)$,
which are included to deal
with the effect of the leader's nonzero control input. When the leader's control input $u_0=0$, the adaptive
protocol \dref{controller3} with the nonlinear components removed reduces to
\begin{equation}\label{controller3q}
\begin{aligned}
\dot{\tilde{v}}_i&=A\tilde{v}_i+Bu_i+F(C\tilde{v}_i-y_i),\\
\dot{\tilde{w}}_i&=Aw_i+Bu_i+(d_i+\rho_i)FC(\tilde{\psi}_i-\tilde{\eta}_i)+F(C\tilde{v}_i-y_i),\\
u_i&=K\tilde{w}_i,\\
\dot{d}_i&=(\psi_i-\tilde{\eta}_i)^TC^TC(\tilde{\psi}_i-\tilde{\eta}_i),\quad i=1,\cdots,N.
\end{aligned}
\end{equation}
Evidently, as a direct consequence of Theorem 3, the adaptive protocol \dref{controller3q} with parameters given as in Theorem 3
can achieve leader-follower consensus for the agents in \dref{model1}
with $u_0=0$, provided that the communication graph satisfies Assumption \ref{assp2}.

For the special case where the relative state information among neighboring agents is available,
we can present the following state feedback adaptive protocol:
\begin{equation}\label{consf2}
\begin{aligned}
u_{i}&=(d_i+\rho_i)K\tilde{\xi}_i-\beta h(B^TP^{-1}\tilde{\xi}_i),\\
\dot{d}_i&=\tilde{\xi}_i^T\Omega\tilde{\xi}_i,\quad  i=1,\cdots,N,
\end{aligned}
\end{equation}
where $\tilde{\xi}_i$ is defined as in \dref{xi}
and $P>0$ is a solution to the LMI \dref{lmib}.

\begin{corollary}\label{corollary1}
Suppose that Assumptions \ref{omega} and \ref{assp2} hold.
Then, the leader-follower consensus problem of the agents in \dref{model1} can be solved by the adaptive protocol \dref{consf2} with $K$, $\Omega$, $\rho_i$ designed as in Theorem \ref{leaderless2} and $\beta\geq\omega$.
\end{corollary}

\begin{proof}
Consider the following Lyapunov function candidate:
\begin{equation}\label{v5}
V_{5}=\sum_{i=1}^{N}\frac{1}{2}g_i(2d_i+\rho_i)\rho_i+\frac{1}{2}\sum_{i=1}^{N}g_i\tilde{d}_i^2,
\end{equation}
where $\lambda_0$ is defined as in \dref{dotlya3} and the rest of the variables are the same as in \dref{lya3}.
Then, by following
similarly steps in the proof of Theorem \ref{thmlf1} and selecting
$\alpha\geq\frac{5\lambda_{\max}(G)}{2\lambda_0}$, we can obtain that the derivative of $V_5$ satisfies
$$\begin{aligned}
\dot{V}_5\leq&-\tilde{\xi}^T[(D+\rho)G\otimes (P^{-1}A+A^TP^{-1}\\&-2P^{-1}BB^TP^{-1})]\tilde{\xi}\\
\leq&0.
\end{aligned}$$
The convergence of the consensus error $\tilde{\xi}$ can be established by the Barbalat's lemma.
The details are omitted here for conciseness. \hfill $\blacksquare$
\end{proof}


\subsection{Continuous Adaptive Consensus Protocols}

Note that the nonlinear function $h(\cdot)$ in the adaptive protocols
\dref{controller3} and \dref{consf2} is discontinuous. In practical implementation,
this discontinuity may result in chattering due to imperfections in switching devices
\cite{young1999control,edwards1998sliding}. 
One feasible
approach to eliminate chattering is to use the boundary layer
technique \cite{young1999control,edwards1998sliding} to
give a continuous approximation of
the discontinuous function $h(\cdot)$; in other words, we may
replace $h(\cdot)$ by a continuous function $\tilde{h}_i(\cdot)$,
defined such that
for $z\in\mathbf{R}^n$,
\begin{equation}\label{satu}
\tilde{h}_i(z)=\begin{cases}\frac{z}{\|z\|} & \text{if}~\|z\|>\kappa_i,\\
\frac{z}{\kappa_i} & \text{if}~\|z\|\leq\kappa_i,
\end{cases}
\end{equation}
where $\kappa_i$ is a small positive constant,
denoting the width of
the boundary layer of the protocol corresponding to the $i$-th follower.
As $\kappa_i\rightarrow 0$,
the continuous function $\tilde{h}_i(\cdot)$
approaches the discontinuous function $h(\cdot)$.
It is worth noting that using this continuous adaptive protocol
with $h(\cdot)$ replaced by $\tilde{h}_i(\cdot)$, the consensus error
$\tilde{\xi}$ will in general not converge to zero
but some small nonzero $\tilde{\xi}$ may result. In this case,
it can be observed from the last equation in \dref{controller3} that
the adaptive gains $d_i(t)$ will slowly grow unbounded.
To tackle this problem, we use
the so-called $\sigma$-modification technique
\cite{ioannou1984instability}
to modify the adaptive protocol \dref{controller3}.

Using the boundary layer concept and the $\sigma$-modification technique,
we propose a new distributed continuous adaptive consensus protocol to each follower as follows:
\begin{equation}\label{controller4}
\begin{aligned}
\dot{\tilde{v}}_i=&A\tilde{v}_i+Bu_i+F(C\tilde{v}_i-y_i),\\
\dot{\tilde{w}}_i=&A\tilde{w}_i+B[u_i-\beta \tilde{h}_i(B^TS(\tilde{\psi}_i-\tilde{\eta}_i))]\\&\quad+(d_i+\rho_i)FC(\tilde{\psi}_i-\tilde{\eta}_i)+F(C\tilde{v}_i-y_i),\\
u_i=&K[\tilde{w}_i-\beta \tilde{h}_i(B^TQ\tilde{\eta}_i)],\\
\dot{d}_i=&-\varphi_i (d_i-1)+(\tilde{\psi}_i-\tilde{\eta}_i)^TC^TC(\tilde{\psi}_i-\tilde{\eta}_i),i=1,\cdots,N,
\end{aligned}
\end{equation}
where $d_i(t)$ denotes the
time-varying coupling weight associated with the $i$-th follower
with $d_i(0)\geq1$,
$\varphi_i$ are small positive constants,
and the rest of the variables are defined as in \dref{controller3}.
Note that $d_i(0)\geq1$ and $\dot{d}_i\geq 0$ when $d_i=1$ in \dref{controller4}.
Then, it is not difficult to see that $d_i(t)\geq 1$ for any $t>0$.

By substituting the adaptive protocol \dref{controller4} into \dref{model1}, we can get the closed-loop dynamics of the network as
\begin{equation}\label{derror22}
\begin{aligned}
\dot{\tilde{\zeta}}=&[I_N\otimes (A+FC)]\tilde{\zeta},\\
\dot{\tilde{\varrho}}=&[I_N\otimes A+\widehat{\mathcal{L}}_1(D+\rho)\otimes FC]\tilde{\varrho}\\
&-(\widehat{\mathcal{L}}_1\otimes B)[\beta \tilde{H}(\tilde{\varrho})-\textbf{1}\otimes u_0]\\
&+(\widehat{\mathcal{L}}_1\otimes FC)(\textbf{1}\otimes e_0),\\
\dot{\tilde{\eta}}=&[I_N\otimes (A+BK)]\tilde{\eta}+(I_N\otimes BK)\tilde{\varrho}\\
&+(I_N\otimes FC)\tilde{\zeta}-(\widehat{\mathcal{L}}_1\otimes B)[\beta \tilde{M}(\tilde{\eta})+\textbf{1}\otimes u_0],\\
\dot{d}_i=&-\varphi_i (d_i-1)+\tilde{\varrho}_i^TC^TC\tilde{\varrho}_i,\quad i=1,\cdots,N,
\end{aligned}
\end{equation}
where $\widetilde{H}(\tilde{\varrho})=[\tilde{h}_1(B^TS\tilde{\varrho}_1)^T,\cdots,\tilde{h}_N(B^TS\tilde{\varrho}_N)^T]^T$,
$\widetilde{M}(\tilde{\eta})=[\tilde{h}_1(B^TQ\tilde{\eta}_1)^T,\cdots,\tilde{h}_N(B^TQ\tilde{\eta}_N)^T]^T$, and the rest of the variables are defined as in \dref{derror2}.

\begin{theorem}\label{thmlf2}
Suppose that Assumptions \ref{omega} and \ref{assp2} hold.
Then, the consensus error $\tilde{\xi}$ of \dref{derror22} and the adaptive gains $d_i,i=1,\cdots,N,$ are uniformly ultimately bounded under the adaptive protocol \dref{controller4} with $\beta$, $K$, $F$, and $Q$ designed as in Theorem \ref{thmlf1}.
Moreover, $\tilde{\xi}$ converges exponentially to the residual set
\begin{equation}\label{set}
\begin{aligned}
\mathcal{D}\triangleq\Bigg\{\tilde{\xi}: \|\tilde{\xi}\|^2<&\frac{\gamma_2(\alpha-1)^2}{2\lambda_{\min}(Q)\delta}\sum_{i=1}^N\varphi_ig_i+\frac{1}{\lambda_{\min}(Q)\delta}\\
\\&\times\sum_{i=1}^N[\gamma_2\Pi_i+(\omega a_{i0}+(2N-1)\beta)\kappa_i]\Bigg\},
\end{aligned}
\end{equation}
where $\alpha$ is defined as in \dref{alpha},
\begin{equation}\label{pi1}
\begin{aligned}
\Pi_i=&g_i(\omega a_{i0}+(2N-1)\beta)\kappa_i\\
&+(\frac{1}{\varphi_i}+\frac{\lambda_{\max}(S)}{2\lambda_{\min}(W)})(\omega a_{i0}+(2N-1)\beta)^2\kappa_i^2g_i.
\end{aligned}
\end{equation}
and
\begin{equation}\label{delta1}
\begin{aligned}
\delta=\min\{\frac{\lambda_{\min}(X)}{2\lambda_{\max}(Q)},\frac{\lambda_{\min}(W)}{\gamma_1\lambda_{\max}(S)},
\frac{\lambda_{\min}(W)}{2\gamma_2\lambda_{\max}(S)},\frac{\varphi_i}{4}\},
\end{aligned}
\end{equation}
with $\gamma_1$ and $\gamma_2$ defined in \dref{gamma12}, $W$ defined in \dref{dotlya3},
$X$ and $\Gamma$ defined in \dref{ineq4}.
\end{theorem}

\begin{proof}
Consider the Lyapunov function $V_3$ as in \dref{lya3}.
The time derivative of $V_{3}$ along \dref{derror22} is given by
\begin{equation}\label{dotlya5}
\begin{aligned}
\dot{V}_3
\leq&\tilde{\varrho}^T[(D+\rho)G\otimes(SA+A^TS)-\lambda_0(D+\rho)^2\otimes C^TC\\
&+(\rho+D-\alpha I_N)G\otimes C^TC]\tilde{\varrho}-\gamma e_0^TWe_0\\
&-\tilde{\varrho}^T[(D+\rho)G\widehat{\mathcal{L}}_1\otimes SB][\beta \widetilde{H}(\tilde{\varrho})-\textbf{1}\otimes u_0]\\
&-2\tilde{\varrho}^T[(D+\rho)G\widehat{\mathcal{L}}_1\otimes C^TC](\textbf{1}\otimes e_0)\\
&-\sum_{i=1}^N\varphi_ig_i(d_i-1)\tilde{d}_i,
\end{aligned}
\end{equation}
where $\lambda_0$ and $W$ are defined as in \dref{dotlya3}.
Note that
\begin{equation}\label{d}
\begin{aligned}
-(d_i-1)\tilde{d}_i=-\tilde{d}_i^2-(\alpha-1)\tilde{d}_i\leq-\frac{1}{2}\tilde{d}_i^2+\frac{1}{2}(\alpha-1)^2,
\end{aligned}
\end{equation}
and
\begin{equation}\label{d1}
\begin{aligned}
-(d_i-1)\tilde{d}_i &=-(d_i-1)^2+(\alpha-1)(d_i-1)\\&\leq-\frac{1}{2}(d_i-1)^2+\frac{1}{2}(\alpha-1)^2.
\end{aligned}
\end{equation}
Substituting \dref{ineq1}, \dref{ineq2}, \dref{ineq3}, \dref{d}, and \dref{d1} into \dref{dotlya5}, we can obtain that
\begin{equation}\label{dotlya52}
\begin{aligned}
\dot{V}_3\leq&-\tilde{\varrho}^T[(D+\rho)G\otimes W]\tilde{\varrho}-e_0^TWe_0\\
&+\frac{(\alpha-1)^2}{2}\sum_{i=1}^N\varphi_ig_i-\sum_{i=1}^N\frac{\varphi_ig_i}{4}[(d_i-1)^2+\tilde{d}_i^2]\\&-\tilde{\varrho}^T[(D+\rho)G\widehat{\mathcal{L}}_1\otimes SB](\beta \widetilde{H}(\tilde{\varrho})-\textbf{1}\otimes u_0).
\end{aligned}
\end{equation}
Consider the following three cases.

i) $\|B^TS\tilde{\varrho}_i\| >\kappa_i, i=1,\cdots,N$.

In this case, as shown in \dref{h}, we can get that
\begin{equation}\label{h1}
\begin{aligned}
&-\tilde{\varrho}^T[(D+\rho)G\widehat{\mathcal{L}}_1\otimes SB]\beta \widetilde{H}(\tilde{\varrho})\\
&\qquad\leq-\sum_{i=1}^N(d_i+\rho_i)g_i\beta a_{i0}\|B^TS\tilde{\varrho}_i\|.
\end{aligned}
\end{equation}
Substituting \dref{u0} and \dref{h1} into \dref{dotlya52} yields
\begin{equation*}
\begin{aligned}
\dot{V}_3\leq&-\tilde{\varrho}^T[(D+\rho)G\otimes W
]\tilde{\varrho}-e_0^TWe_0\\
&-\sum_{i=1}^N\frac{\varphi_ig_i}{4}[(d_i-1)^2+\tilde{d}_i^2]+\frac{(\alpha-1)^2}{2}\sum_{i=1}^N\varphi_ig_i.
\end{aligned}
\end{equation*}

ii) $\|B^TS\tilde{\varrho}_i\| \leq\kappa_i$, $i=1,\cdots,N$.

In this case, it follows from \dref{u0} and \dref{satu} that
\begin{equation}\label{h2}
\begin{aligned}
-\tilde{\varrho}^T&[(D+\rho)G\widehat{\mathcal{L}}_1\otimes SB](\beta \widetilde{H}(\tilde{\varrho})-\textbf{1}\otimes u_0)\\
\leq&\sum_{i=1}^N(d_i+\rho_i)g_i(\omega a_{i0}+(2N-1)\beta)\|B^TS\tilde{\varrho}_i\|\\
\leq&\sum_{i=1}^N(d_i+\rho_i)g_i(\omega a_{i0}+(2N-1)\beta)\kappa_i\\
\leq&\sum_{i=1}^N\frac{\varphi_ig_i}{4}(d_i-1)^2+\sum_{i=1}^N\frac{\lambda_{\min}(W)}{2\lambda_{\max}(S)}g_i\rho_i^2+\sum_{i=1}^N\Pi_i\\
\leq&\sum_{i=1}^N\frac{\varphi_ig_i}{4}(d_i-1)^2+\frac{1}{2}\tilde{\varrho}^T(\rho G\otimes W)\tilde{\varrho}+\sum_{i=1}^N\Pi_i,
\end{aligned}
\end{equation}
where we have used Lemma \ref{ineq} to get the third inequality and
$\Pi_i$ is defined as in \dref{pi1}.
Substituting \dref{h2} into \dref{dotlya52} yields
\begin{equation}\label{dotlya53}
\begin{aligned}
\dot{V}_3\leq&-\frac{1}{2}\tilde{\varrho}^T[(D+\rho)G\otimes W]\tilde{\varrho}-e_0^TWe_0\\&-\sum_{i=1}^N\frac{\varphi_ig_i}{4}\tilde{d}_i^2+\frac{(\alpha-1)^2}{2}\sum_{i=1}^N\varphi_ig_i+\sum_{i=1}^N\Pi_i.
\end{aligned}
\end{equation}

iii) $\tilde{\varrho}$ satisfies neither Case i) nor Case ii).

Without loss of generality, assume that $\|B^TS\tilde{\varrho}_i\| >\kappa_i$, $i=1,\cdots,l$, and $\|B^TS\tilde{\varrho}_i\| \leq\kappa_i$, $i=l+1,\cdots,N$, where $2\leq l\leq N-1$. By combining \dref{h1} and \dref{h2}, in this case we can get that
\begin{equation}\label{h3}
\begin{aligned}
-\tilde{\varrho}^T&[(D+\rho)G\widehat{\mathcal{L}}_1\otimes SB](\beta \widetilde{H}(\tilde{\varrho})-\textbf{1}\otimes u_0)\\
&\leq\sum_{i=l+1}^N[\frac{\varphi_ig_i}{4}(d_i-1)^2+ \frac{\lambda_{\min}(W)}{2\lambda_{\max}(S)}g_i\rho_i^2+\Pi_i].
\end{aligned}
\end{equation}
Then, it follows from \dref{dotlya52} and \dref{h3} that
\begin{equation*}
\begin{aligned}
\dot{V}_3\leq&-\frac{1}{2}\tilde{\varrho}^T[(D+\rho)G\otimes W]\tilde{\varrho}-e_0^TWe_0\\&-\sum_{i=1}^N\frac{\varphi_ig_i}{4}\tilde{d}_i^2+\frac{(\alpha-1)^2}{2}\sum_{i=1}^N\varphi_ig_i
+\sum_{i=l+1}^N\Pi_i.
\end{aligned}
\end{equation*}

Therefore, by analyzing the above three cases, we get that $\dot{V}_3$ satisfies \dref{dotlya53} for all $\tilde{\varrho}\in\mathbf{R}^{Nn}$.

Next, consider the Lyapunov function $V_4$ in \dref{lya4}.
Similar to the discussion of three cases above, we have that
\begin{equation}\label{kappa}
\begin{aligned}
&-2\tilde{\eta}^T(\widehat{\mathcal{L}}_1\otimes QB)(\beta \widetilde{M}(\tilde{\eta})+\textbf{1}\otimes u_0)\\
&\quad\qquad\leq \sum_{i=1}^N[\omega a_{i0}+(2N-1)\beta]\kappa_i.
\end{aligned}
\end{equation}
By using \dref{ineq4}, \dref{dotlya53}, and \dref{kappa}, we can obtain
the time derivative of $V_{4}$ along \dref{derror2} as
\begin{equation}\label{dotlya61}
\begin{aligned}
\dot{V}_4\leq&-\frac{1}{2}\tilde{\eta}^T(I_N\otimes X)\tilde{\eta}-\tilde{\zeta}^T(I_N\otimes W)\tilde{\zeta}\\
&-\frac{1}{2}\tilde{\varrho}^T[(D+\rho)G\otimes W]\tilde{\varrho}-\gamma_2e_0^TWe_0\\&-\gamma_2\sum_{i=1}^N\frac{\varphi_ig_i}{4}\tilde{d}_i^2+\frac{\gamma_2(\alpha-1)^2}{2}\sum_{i=1}^N\varphi_ig_i
\\
&+\sum_{i=1}^N[\gamma_2\Pi_i+(\omega a_{i0}+(2N-1)\beta)\kappa_i].
\end{aligned}
\end{equation}
Furthermore, we rewrite \dref{dotlya61} into
\begin{equation}\label{dotlya62}
\begin{aligned}
\dot{V}_4\leq&-\delta V_4+\delta V_4-\frac{1}{2}\tilde{\eta}^T(I_N\otimes X)\tilde{\eta}-\tilde{\zeta}^T(I_N\otimes W)\tilde{\zeta}\\
&-\frac{1}{2}\tilde{\varrho}^T[(D+\rho)G\otimes W]\tilde{\varrho}-\gamma_2e_0^TWe_0-\gamma_2\sum_{i=1}^N\frac{\varphi_ig_i}{4}\tilde{d}_i^2\\
&+\sum_{i=1}^N[\frac{\gamma_2(\alpha-1)^2}{2}\varphi_ig_i+\gamma_2\Pi_i+(\omega a_{i0}+(2N-1)\beta)\kappa_i]\\
\leq&-\frac{1}{2}\tilde{\eta}^T[I_N\otimes (X-2\delta Q)]\tilde{\eta}-\tilde{\zeta}^T[I_N\otimes (W-\gamma_1\delta S)]\tilde{\zeta}\\
&-\frac{1}{2}\tilde{\varrho}^T[(D+\rho)G\otimes (W
-2\gamma_2\delta S)]\tilde{\varrho}-\gamma_2e_0^T(W-\delta S)e_0\\&-\gamma_2\sum_{i=1}^N(\frac{\varphi_i}{4}-\delta)g_i\tilde{d}_i^2+\frac{\gamma_2(\alpha-1)^2}{2}\sum_{i=1}^N\varphi_ig_i
\\
&+\sum_{i=1}^N[\gamma_2\Pi_i+(\omega a_{i0}+(2N-1)\beta)\kappa_i]-\delta V_4,
\end{aligned}
\end{equation}
where $\delta$ is defined as in \dref{delta1}.
By the definition of $\delta$, it follows from \dref{dotlya62} that
\begin{equation}\label{dotlya63}
\begin{aligned}
\dot{V}_4\leq&-\delta V_4+\frac{\gamma_2(\alpha-1)^2}{2}\sum_{i=1}^N\varphi_ig_i
\\&+\sum_{i=1}^N[\gamma_2\Pi_i+(\omega a_{i0}+(2N-1)\beta)\kappa_i].
\end{aligned}
\end{equation}
In light of Lemma \ref{comparison}, we can deduce from \dref{dotlya63} that 
$V_6$ exponentially converges to the residual set
\begin{equation}\label{D1}
\begin{aligned}
\mathcal{D}_1\triangleq\Bigg\{V_4: V_4<&\frac{\gamma_2(\alpha-1)^2}{2\delta}\sum_{i=1}^N\varphi_ig_i
\\&+\frac{1}{\delta}\sum_{i=1}^N[\gamma_2\Pi_i+(\omega a_{i0}+(2N-1)\beta)\kappa_i]\Bigg\}
\end{aligned}
\end{equation}
with a convergence rate faster than $e^{-\delta t}$. Since  $V_4\geq\lambda_{\min}(Q)\|\tilde{\eta}\|^2+\min g_i\,\gamma_2(\lambda_{\min}(S)\|\tilde{\varrho}\|^2+\frac{1}{2}\sum_{i=1}^N\tilde{d}_i^2)$,
it follows from \dref{D1} that $\tilde{\eta}$, $\tilde{\varrho}$, and
$d_i$ are uniformly ultimately bounded.
Moreover, since $\tilde{\zeta}$ asymptotically converges to zero, we can obtain that
$\tilde{\eta}$ and thereby $\tilde{\xi}$ exponentially converges to the residual set $\mathcal{D}$
with a convergence rate faster than $e^{-\delta t}$.\hfill $\blacksquare$
%
\end{proof}

For the special case where the relative states among neighboring agents are available,
the discontinuous state feedback adaptive protocol \dref{consf2} can be modified to be
\begin{equation}\label{consf2c}
\begin{aligned}
u_{i}&=(d_i+\rho_i)K\tilde{\xi}_i-\beta \tilde{h}_i(B^TP^{-1}\tilde{\xi}_i),\\
\dot{d}_i&=-\varphi_i (d_i-1)+\tilde{\xi}_i^T\Omega\tilde{\xi}_i,\quad  i=1,\cdots,N,
\end{aligned}
\end{equation}
where $\tilde{h}_i(\cdot)$ is defined as in \dref{satu} and $\varphi_i$ are small positive constants.

\begin{corollary}\label{corollary2}
Suppose that Assumptions \ref{omega} and \ref{assp2} hold.
Then, the consensus error $\tilde{\xi}$ and the adaptive gains $d_i,i=1,\cdots,N,$ are uniformly ultimately bounded under
the adaptive protocol \dref{consf2c} with $K$, $\Omega$, $\rho_i$, and $\beta$ designed as in Corollary \ref{corollary1}.
Moreover,
$\tilde{\xi}$ converges exponentially to the residual set
\begin{equation}\label{set1}
\begin{aligned}
\tilde{\mathcal{D}}\triangleq\Bigg\{\tilde{\xi}: \|\tilde{\xi}\|^2<&\frac{1}{\lambda_{\min}(P^{-1})\tilde{\delta}}\sum_{i=1}^N[\frac{1}{2}(\alpha-1)^2\varphi_ig_i+\tilde{\Pi}_i]\Bigg\},
\end{aligned}
\end{equation}
where $\alpha$ is defined as in \dref{v5},
$\tilde{\delta}=\min\{\frac{\lambda_{\min}(\tilde{X})}{2\lambda_{\max}(P^{-1})},\frac{\varphi_i}{4}\}$,
$$
\begin{aligned}
\tilde{\Pi}_i=&g_i(\omega a_{i0}+(2N-1)\beta)\kappa_i\\
&+(\frac{1}{\varphi_i}+\frac{\lambda_{\max}(P^{-1})}{2\lambda_{\min}(\tilde{X})})(\omega a_{i0}+(2N-1)\beta)^2\kappa_i^2g_i.
\end{aligned}
$$
with $P$ defined in \dref{lmib}, $\tilde{X}=-(P^{-1}A+A^TP^{-1}-2\Omega)$,
and $\Omega$ defined in Theorem \ref{leaderless2}.
\end{corollary}

\begin{proof}
Consider the Lyapunov function $V_{5}$ in \dref{v5}. By following
similarly steps in the proof of Theorem \ref{thmlf2},
we can obtain that
\begin{equation}\label{dotlya532}
\begin{aligned}
\dot{V}_5\leq&-\frac{1}{2}\tilde{\xi}^T[(D+\rho)G\otimes \tilde{X}]\tilde{\xi}-\sum_{i=1}^N\frac{\varphi_ig_i}{4}\tilde{d}_i^2\\&+\frac{(\alpha-1)^2}{2}\sum_{i=1}^N\varphi_ig_i+\sum_{i=1}^N\tilde{\Pi}_i.
\end{aligned}
\end{equation}

The upper bound of the consensus error $\tilde{\xi}$ can be obtained
by following the last part of the proof of Theorem \ref{thmlf2}.
The details are omitted here for conciseness. \hfill $\blacksquare$
\end{proof}

\begin{remark}
It is worth mentioning that
implementing the $\sigma$-modification technique
to add $-\varphi_i (d_i-1)$ into \dref{controller4} or \dref{consf2c}
and using the boundary layer concept to derive continuous functions $\tilde{h}_i$
play a vital role to guarantee the ultimate
boundedness of the consensus error $\tilde{\xi}$
and the adaptive gains $d_i$.
We can observe from \dref{set} and \dref{set1} that
the upper bounds of the consensus error $\tilde{\xi}$ depend on the $\sigma$-modification parameters $\varphi_i$
and the boundary layer widths $\kappa_i$. In practice, we can choose $\varphi_i$ and $\kappa_i$
to be relatively small in order to guarantee
a small consensus error $\xi$.
\end{remark}

\begin{remark}
Compared to the previous related works \cite{cao2011distributed,li2012adaptiveauto,li2011trackingTAC}, which are applicable to
only undirected subgraphs among the followers, the results in this section solve
the distributed tracking problem in the presence of a leader with nonzero control input for general directed graphs.
It should be noted that even though no global information of the communication graph is
needed in the adaptive protocols,
the upper bound of the leader's control input is nonetheless required.
This latter requirement appears to be a limitation of the present adaptive protocols, albeit a modest one.
\end{remark}

\section{Simulation Examples}\label{s5}

In this section, we present numerical simulations to illustrate the effectiveness of the preceding theoretical results.

{\it Example 1}:
Consider a network of second-order integrators, described by \dref{model1}, with
$$x_i=\begin{bmatrix}
x_{i1}\\
x_{i2}\end{bmatrix}, ~ A=\begin{bmatrix}
0\quad 1\\
0\quad 0
\end{bmatrix}, ~B=\begin{bmatrix}
0\\
1
\end{bmatrix}, ~C=\begin{bmatrix}
1\quad 0
\end{bmatrix}.$$
The communication graph is given as in Fig. \ref{figgraph1}, which is strongly connected.

It is worth noting that for second-order integrators with directed graphs,
determining the parameters in existing linear consensus protocols
generally requires the Lapalican matrix's nonzero eigenvalues \cite{ren2008consensus,tuna2009conditions}.
Therefore, we will use the adaptive protocol \dref{controller1} to solve the consensus problem.
\begin{figure}[htbp]
\begin{center}
\includegraphics[width=2.0in]{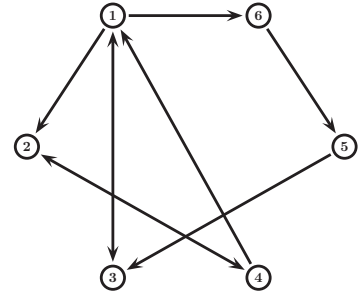}    
\caption{A strongly connected directed communication graph.}  
\label{figgraph1}                                 
\end{center}                                 
\end{figure}

Select $K=-\begin{bmatrix}0.8543 &  2.5628\end{bmatrix}$ such that $A+BK$ is Hurwitz.
Solving the LMI \dref{lmic} by using the LMI toolbox of Matlab, we obtain a feasible solution
$S=\begin{bmatrix}
0.5853 & -0.5853\\
-0.5853& 1.7559
\end{bmatrix}.$ The feedback gain matrix of \dref{controller1} is given by
$F=-\begin{bmatrix}2.5628\\0.8543  \end{bmatrix}$.
Let $d_i(0)=1$, $i=1,\cdots,6$. Then, with the adaptive protocol \dref{controller1}, the relative states $x_i-x_1$, $i=2,\cdots,6,$ of
the second-order integrators are
depicted in Fig. \ref{figerror1}. Evidently,
consensus is achieved. The adaptive coupling weights $d_i$ in \dref{controller1}
are shown in Fig. \ref{figdi10}, which converge to finite steady-state values.

\begin{figure}[htbp]
  \centering
  \includegraphics[width=2.3in]{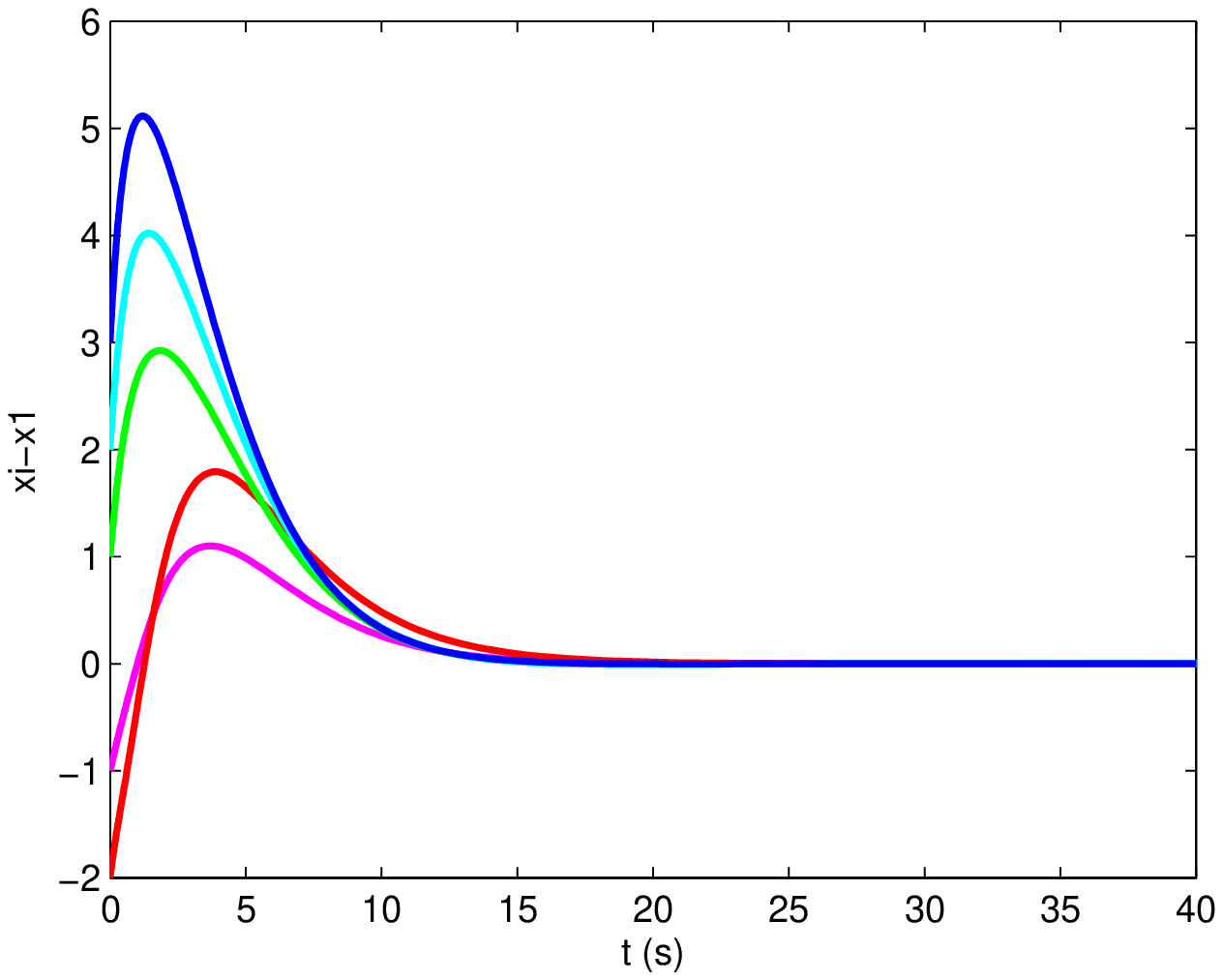}\\
  \caption{The consensus errors $x_i-x_1$, $i=2,\cdots,6$ under \dref{controller1}.}\label{figerror1}
\end{figure}

\begin{figure}[htbp]
\centering
\includegraphics[width=2.3in]{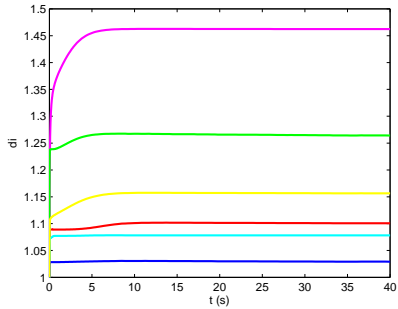}\\
\caption{The adaptive gains $d_i$ in \dref{controller1}.}\label{figdi10}
\end{figure}

\begin{figure}[htbp]
\begin{center}
\includegraphics[width=2.0in]{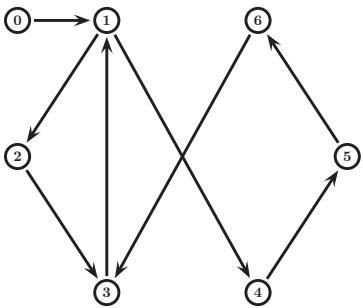}    
\caption{A leader-follower graph containing a directed spanning tree.}  
\label{graphlfd1}                                 
\end{center}                                 
\end{figure}

{\it Example 2}: Consider a network of heterogeneous agents consisting of a leader and several followers.
Let the leader be a nonlinear Chua's circuit, whose dynamics in the dimensionless
form are given by \cite{madan1993chua}
\begin{equation}\label{chua2}
\begin{aligned}
\dot{x}_{0} &= Ax_{0}+Bf_0(x_{0}),
\end{aligned}
\end{equation}
where
$$\begin{aligned}
x_0 &=\left[\begin{matrix} x_{01}\\x_{02}\\x_{03}\end{matrix}\right], A =\left[\begin{matrix}-a(m_0^1+1) & a & 0\\
1 & -1 &1 \\ 0 & -b & 0\end{matrix}\right], B=\left[\begin{matrix} 1 \\ 0 \\ 0\end{matrix}\right],\\
f_0(x_{0}) &=\frac{a}{2}(m_0^1-m_0^2)(|x_{01}+1|-|x_{01}-1|),
\end{aligned}$$
with $a>0$, $b>0$, $m_0^1<0$, and $m_0^2<0$ being the parameters of Chua's circuits.
Let $a=9$, $b= 18$, $m_0^1=-\frac{3}{4}$, and $m_0^2=-\frac{4}{3}$.
In this case, the leader displays a
double-scroll chaotic attractor \cite{madan1993chua}. The nonlinear term $f_0(x_{0})$ in \dref{chua2} is regarded as
the control input of the leader, which is bounded.
The followers are three-order linear systems, described by \dref{model1},
with $A$ and $B$ given in \dref{chua2}.

The communication graph among the agents are depicted in Fig. \ref{graphlfd1}, where the node indexed by
0 is the leader.
For simplicity, it is assumed that the relative state information of neighboring agents is available
and the continuous adaptive protocol \dref{consf2c} is used to achieve leader-follower consensus.
Solving the linear matrix inequality \dref{lmib} gives
$P=\left[\begin{matrix} 0.2403 &  -0.1467 &  -0.3444\\
   -0.1467  &  0.1459 &   0.0332\\
   -0.3444  &  0.0332 &   2.8821\end{matrix}\right].$
The feedback gain matrices of \dref{consf2c}, accordingly, is obtained as
$K=-\left[\begin{matrix} 2.8843  & 3.1711  &1.5114 \end{matrix}\right]$ and
$\Gamma=\left[\begin{matrix} 8.3194  &  9.1465  &  4.3594\\
    9.1465 &  10.0558  &  4.7928\\
    4.3594  &  4.7928  &  2.2843 \end{matrix}\right].$
The initial state of the leader is chosen as $x_0(0)=[1,0.8,-1.5]^T$
and the initial states of the agents are randomly chosen.
Select $\beta=10$, $\kappa=0.05$, and $\varphi_i=0.02$ in \dref{consf2c}.
The state trajectories $x_i(t)$ of the agents under
\dref{consf2c}, designed as above, are depicted in Fig. \ref{fig5},
demonstrating that leader-follower consensus
is indeed achieved.
The adaptive gains
$d_{i}$
in \dref{controller4} are shown in Fig. \ref{figdi1}, which
are clearly bounded.


\begin{figure}[htbp]\centering
\includegraphics[width=0.47\linewidth,height=0.35\linewidth]{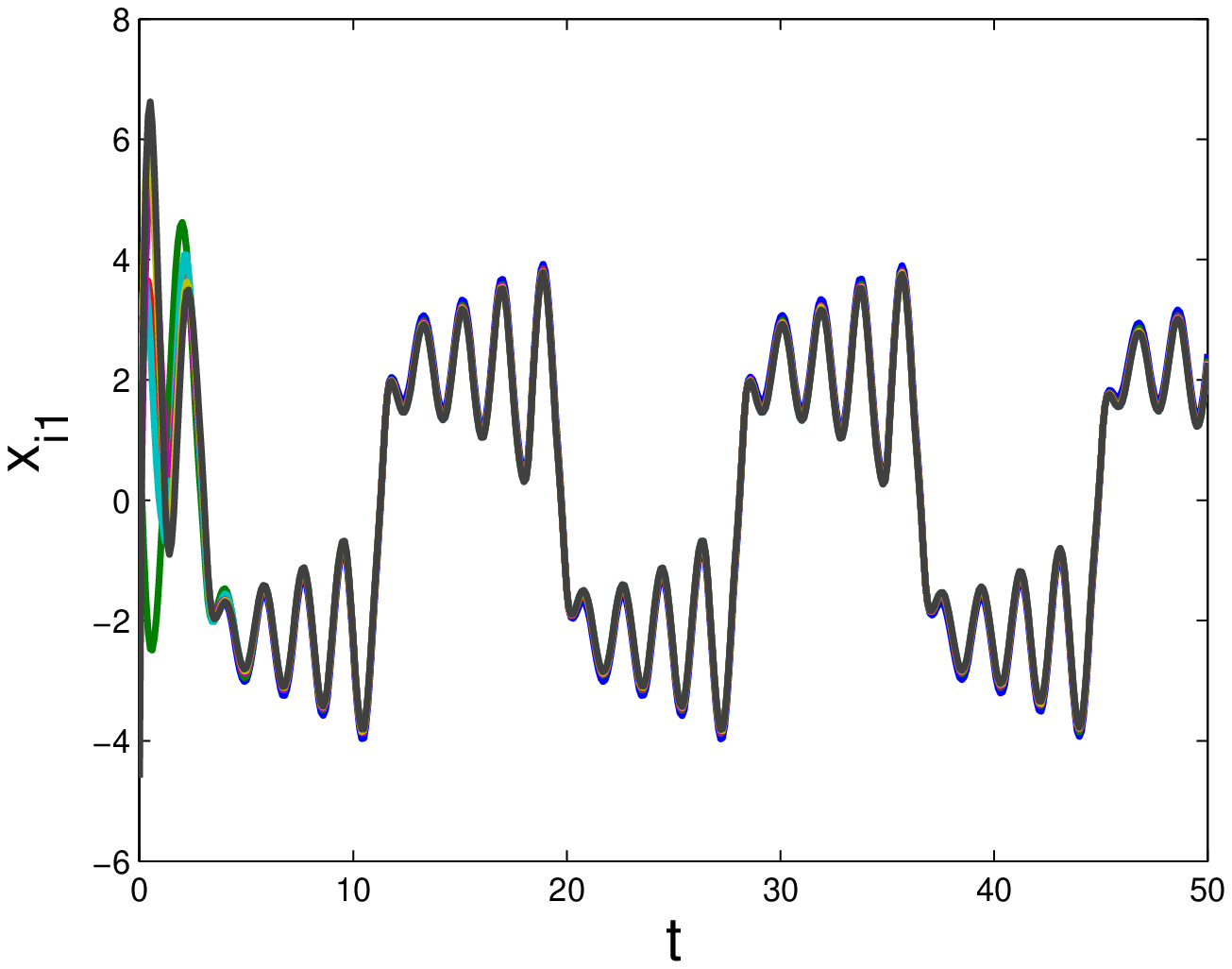}~
\includegraphics[width=0.47\linewidth,height=0.35\linewidth]{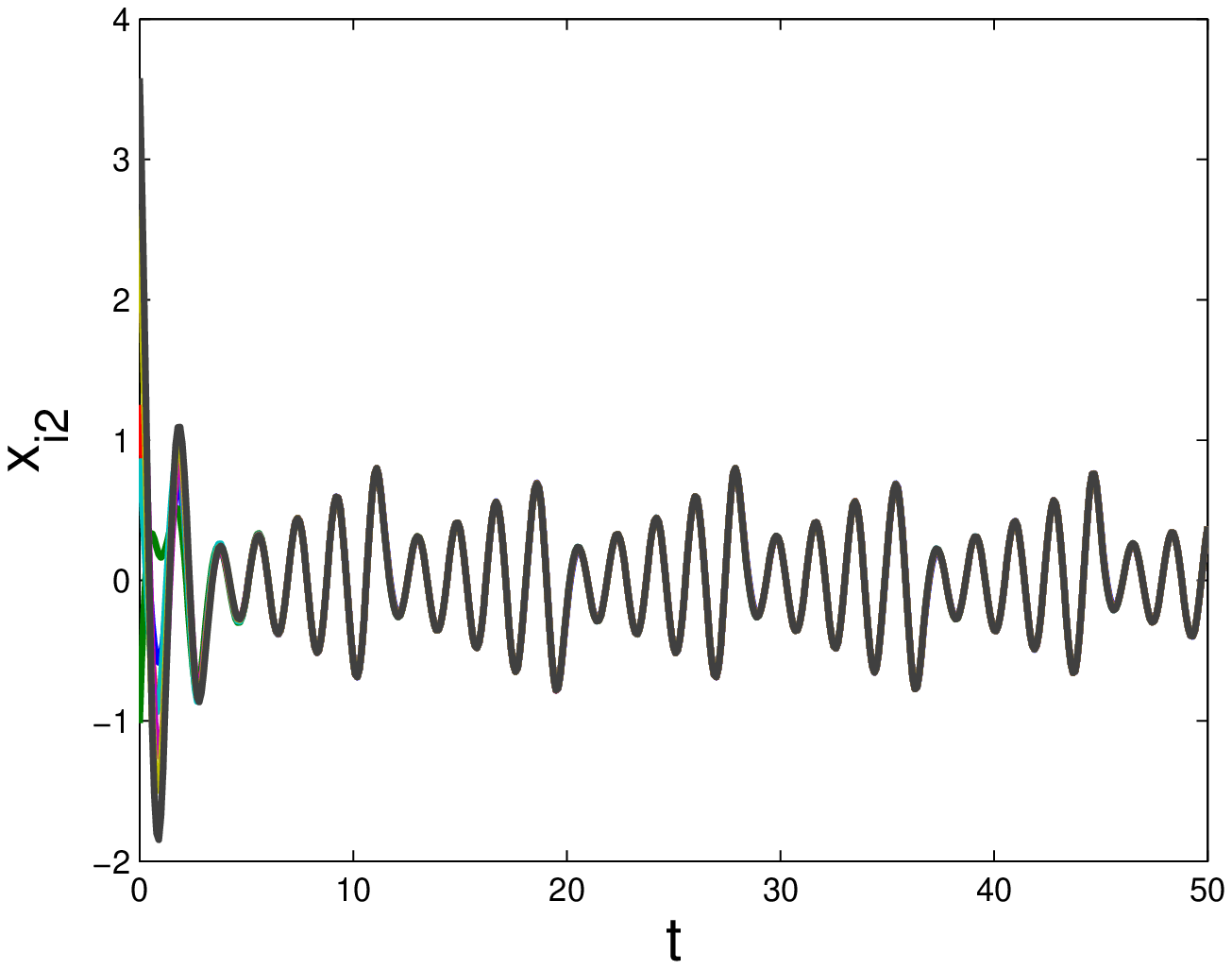}\\
\includegraphics[width=0.47\linewidth,height=0.35\linewidth]{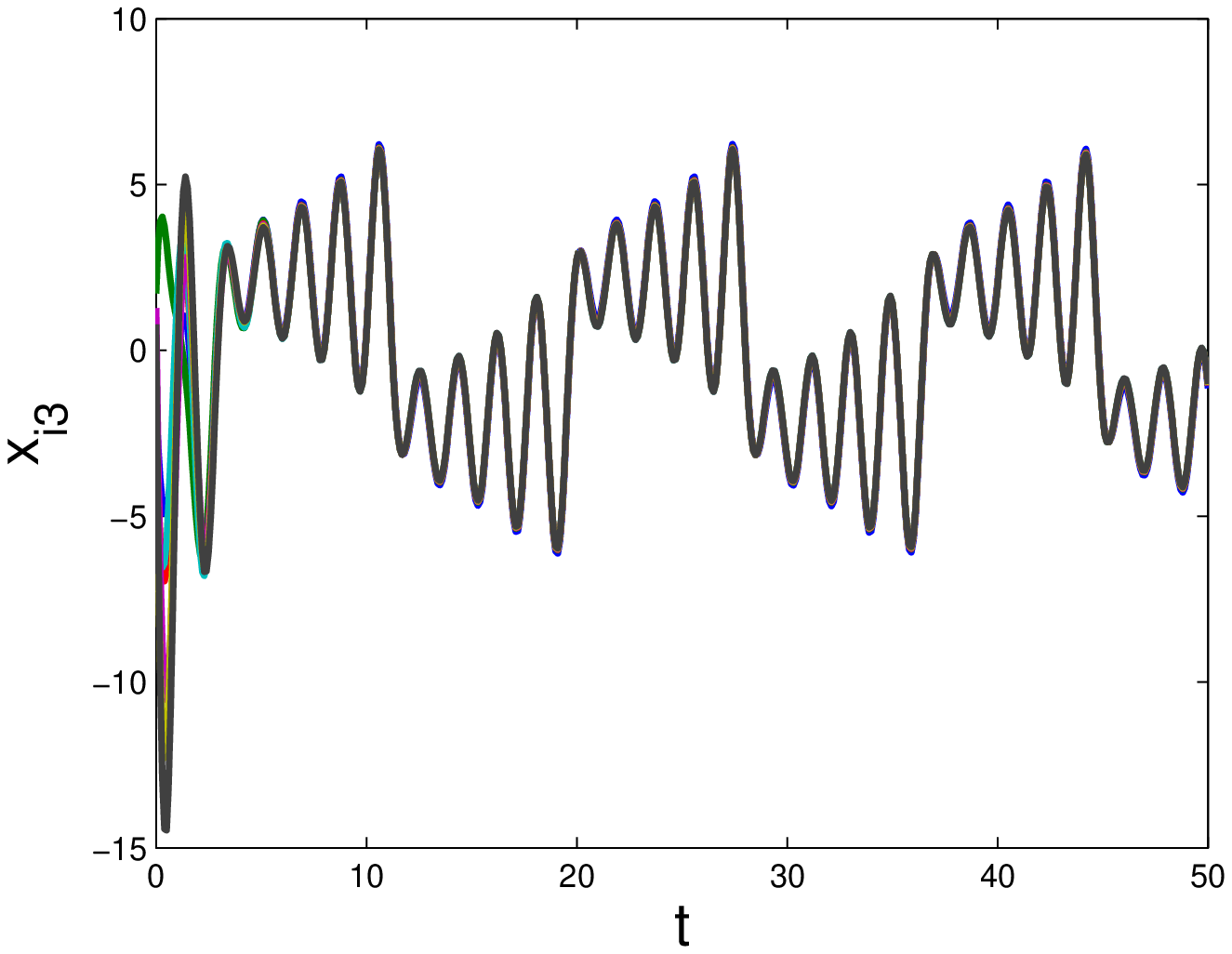}
\caption{The state trajectories of the leader and the followers under the adaptive protocol \dref{consf2c}. }\label{fig5}
\end{figure}

\begin{figure}[htbp]
\centering
\includegraphics[width=2.3in]{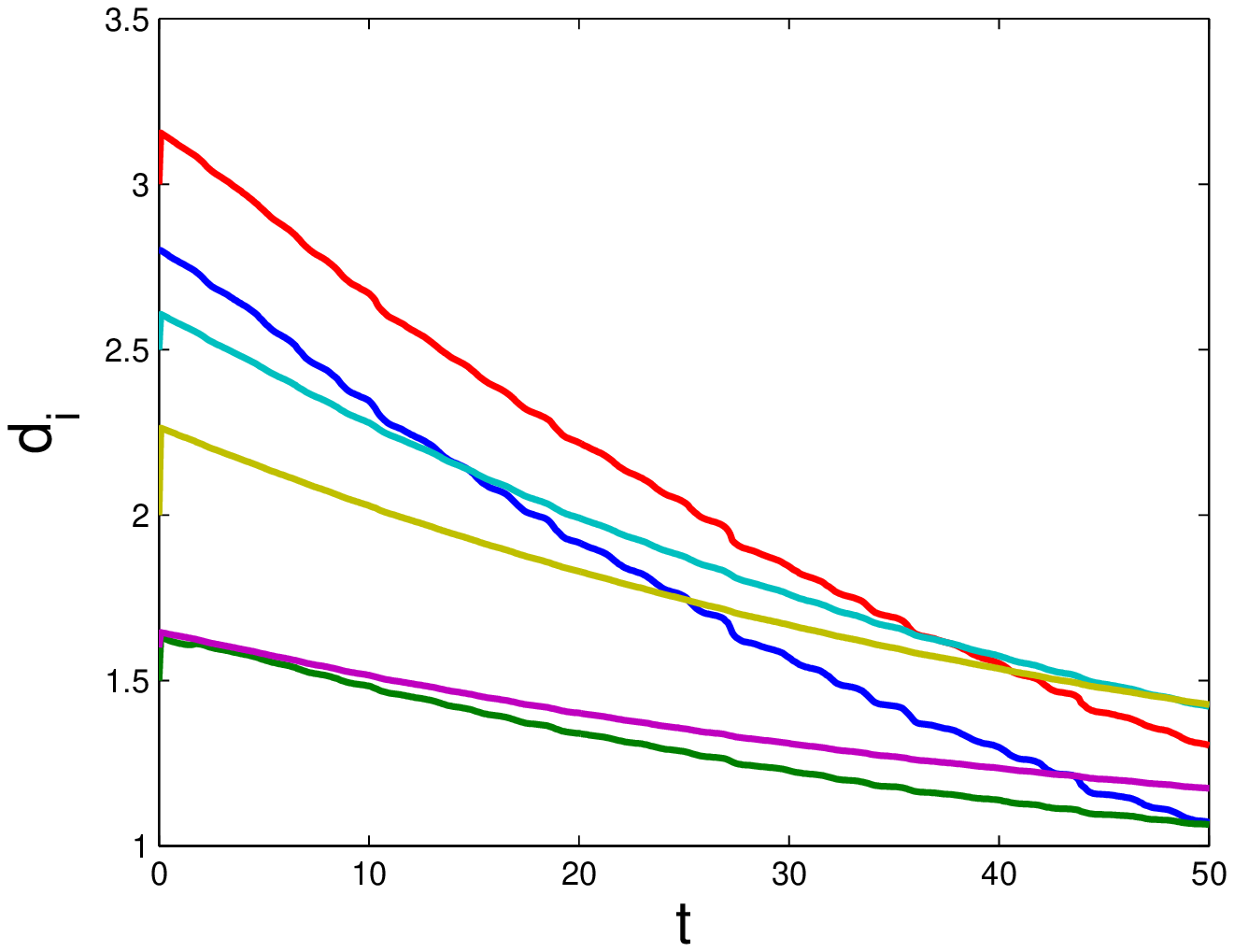}\\
\caption{The adaptive gains $d_i$ in \dref{consf2c}.}\label{figdi1}
\end{figure}

\section{Conclusion}\label{s6}

In this paper, we have addressed the distributed output feedback consensus protocol design problem for linear multi-agent systems with directed graph. One main contribution of this paper is that a new SOD method has been introduced to derive distributed adaptive output feedback consensus protocols, which can solve the leaderless consensus problem for linear multi-agent systems with general directed graphs and the leader-follower consensus problem for the case  with a leader of bounded control input.

The adaptive output feedback protocols for the leaderless case are independent of any global information of the communication graph,
and thereby are fully distributed. It should be mentioned that the adaptive output feedback protocols for the leader-follower case require
the upper bound of the leader's control input. 
This issue will be addressed in our future works.
Another future direction is to consider the case where the agents are non-introspective, i.e.,
having access to only the relative output information respect to their neighbors.


\end{document}